\documentclass[sigconf]{acmart}

\usepackage{algorithm}
\usepackage{algorithmic}
\usepackage{multirow}

\newtheorem{assumption}{Assumption}
\newtheorem{property}{Property}
\newtheorem{definition}{Definition}
\newtheorem{theorem}{Theorem}

\acmYear{2026}
\acmConference[DAI '26]
  {The Eighth International Conference on Distributed Artificial Intelligence}
  {November 29--December 2, 2026}
  {Hong Kong, China}
\acmBooktitle{The Eighth International Conference on Distributed Artificial
  Intelligence (DAI '26), November 29--December 2, 2026, Hong Kong, China}

\setcopyright{none}
\copyrightyear{}
\acmDOI{}
\acmISBN{}

\begin{document}

\title[CARVS: A Congestion-Aware Algorithm for Online MAPD in Human-Collaborative Warehouses]{CARVS: A Congestion-Aware Algorithm for Online Multi-Agent Pickup-and-Delivery in Human-Collaborative Warehouses}

\author{Chaoran Li}
\affiliation{%
  \institution{Shanghai Jiao Tong University}
  \department{Department of Computer Science and Engineering}
  \city{Shanghai}
  \country{China}
}
\email{seraphli@sjtu.edu.cn}

\author{Zeyuan Zhao}
\affiliation{%
  \institution{Shanghai Jiao Tong University}
  \department{Department of Computer Science and Engineering}
  \city{Shanghai}
  \country{China}
}

\author{Shao Zhang}
\affiliation{%
  \institution{Shanghai Artificial Intelligence Laboratory}
  \city{Shanghai}
  \country{China}
}

\renewcommand{\shortauthors}{Li et al.}

\begin{abstract}
Intelligent warehouse systems have transformed logistics by integrating robots with minimal infrastructure changes, boosting efficiency and optimizing workflows compared to traditional human-only operations. However, human-robot collaboration introduces new challenges, particularly due to human uncertainty, which leads to dynamic interactions and unpredictable behavior. To address these complexities, this paper presents the Multi-Agent Pickup-and-Delivery with Human Collaboration (MAPDwHC) problem, an extension of the Multi-Agent Pickup-and-Delivery (MAPD) problem that incorporates human unpredictability. This problem adds key constraints, such as \emph{dynamic obstacles} caused by robots temporarily blocking paths when interacting with humans, and the need for \emph{real-time responsiveness}. We formalize the MAPDwHC problem with additional constraints and properties. Based on this problem setting, we propose the CARVS algorithm, which improves task assignment and path planning by integrating task swapping, validation checks, and random initialization. CARVS reduces congestion, enhances adaptability to human-robot interactions, and ensures real-time responsiveness. Simulation results show that CARVS significantly reduces timeouts of path planning in large-scale scenarios where traditional MAPD algorithms fail, improving both efficiency and scalability.
\end{abstract}



\maketitle

\section{Introduction}

The rise of automation in logistics, particularly in warehouse management, has driven the development of advanced multi-agent systems \cite{kouIdleTimeOptimization2020,zhangCollisionfreeRoutePlanning2018}.
These systems are crucial for enhancing efficiency and managing the complexity of automated warehouses \cite{duinkerkenRoutingAGVsAutomated2015}.
In particular, Multi-Agent Path Finding (MAPF) \cite{sternMultiAgentPathfindingDefinitions2019} plays a key role in coordinating robot movements to prevent conflicts and optimize path efficiency \cite{maMultiAgentPathFinding2016}.
While significant progress has been made in MAPF research \cite{liAnytimeMultiAgentPath2021,shenTrackingProgressMultiAgent2023}, the integration of task assignment within this framework remains underexplored.
To address this, Multi-Agent Pickup and Delivery (MAPD) algorithms \cite{maLifelongMultiagentPath2017} were developed, combining task assignment and pathfinding to enhance system adaptability and throughput.
Since MAPF is NP-hard \cite{maMultiAgentPathFinding2016} and MAPD generalizes MAPF, MAPD is also NP-hard, complicating system-wide optimization. To reduce computational complexity, traditional approaches often separate task assignment from path planning \cite{chenIntegratedTaskAssignment2021}.
While A*-based algorithms \cite{Hart1968Astar,gongSafeEfficientHuman2023} optimize task assignment by minimizing travel time, they frequently overlook warehouse congestion, agent status, and environmental constraints, resulting in poor real-time adaptability and delayed responses to dynamic changes.

\begin{figure}[!t]
    \centering
    \includegraphics[width=\columnwidth]{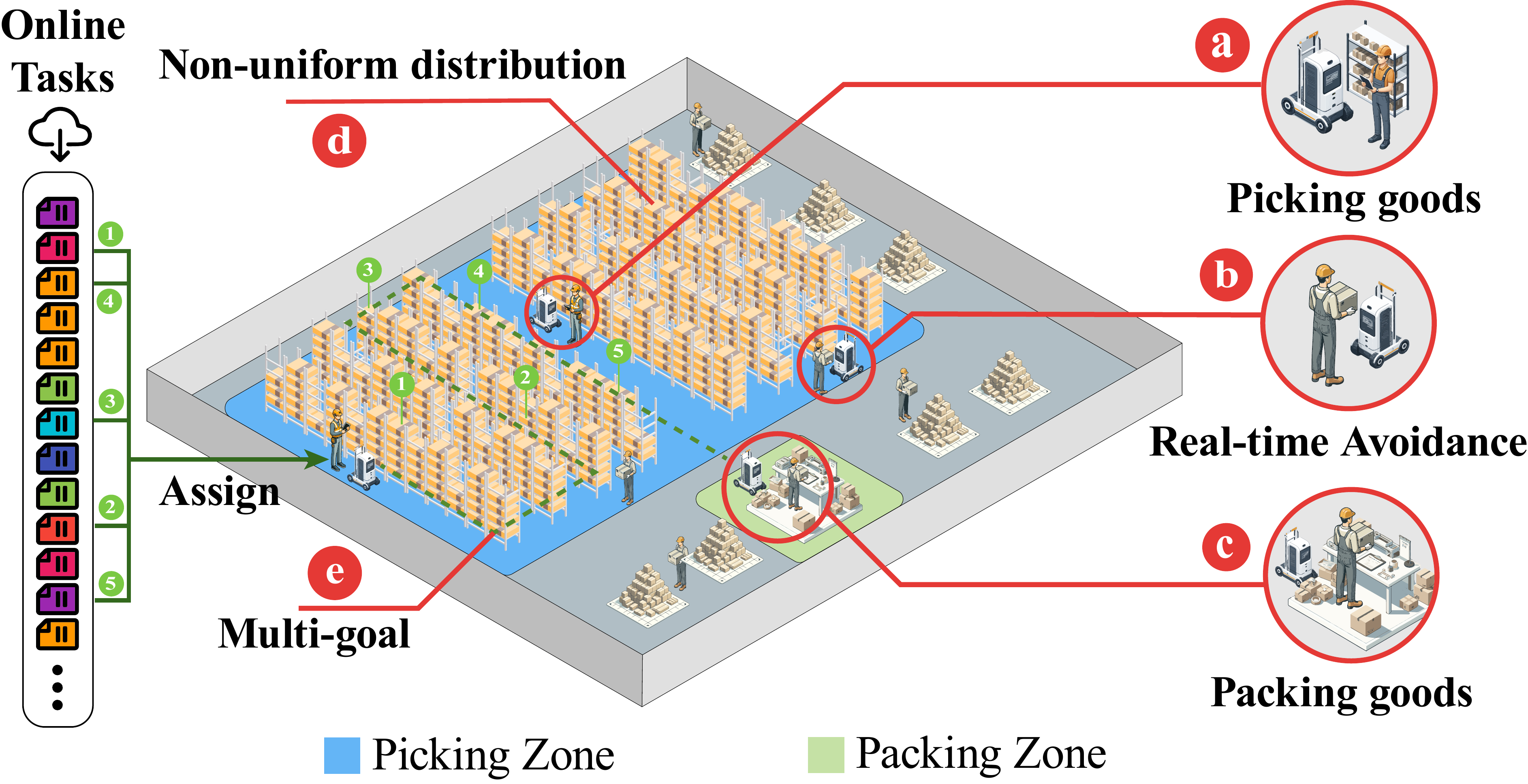}
    \caption{An illustration of the online MAPDwHC problem, showing key features: (a, c) Robots must stay stationary during human interaction, creating \emph{dynamic obstacles}. (b) Robots must avoid humans and obstacles, necessitating algorithms with \emph{real-time responsiveness}. (d) Items are unevenly distributed. (e) Robots can carry multiple items.}
    \label{fig:problem}
    \Description{An illustration of the online MAPDwHC problem, showing key features: (a, c) Robots must stay stationary during human interaction, creating \emph{dynamic obstacles}. (b) Robots must avoid humans and obstacles, necessitating algorithms with \emph{real-time responsiveness}. (d) Items are unevenly distributed. (e) Robots can carry multiple items.}
\end{figure}

Despite these advancements, MAPD assumes an environment free of human workers, limiting its real-world applicability.
Modern intelligent warehouse systems \cite{gongSafeEfficientHuman2023} have transformed logistics by integrating Autonomous Mobile Robots (AMRs) \cite{loganathanSystematicReviewRecent2023} without requiring major infrastructure changes.
Full warehouse automation requires standardized shelving and uniform packaging, which is prohibitively expensive for most existing facilities.
While AMRs handle routine transport tasks, human workers perform critical operations such as picking and packing, which require dexterity to handle items of varying sizes.
During these operations, robots must remain stationary while interacting with humans, temporarily forming \textbf{dynamic obstacles} (see Figure~\ref{fig:problem}).
Warehouse aisles are typically narrow to maximize storage capacity, so a stationary agent during interaction can block paths for other agents, rendering some MAPF instances unsolvable.
Unlike robots, humans do not strictly follow predefined schedules, resulting in unpredictable behavior that affects task allocation and path planning.
Consequently, algorithms need \textbf{real-time responsiveness} to respond dynamically to these changing conditions.
These constraints significantly heighten the problem's challenge and necessitate adaptive algorithms capable of reacting to real-time changes.

Existing MAPD algorithms are insufficient for addressing these challenges, as they assume static environments without human involvement.
Adding human collaboration imposes time-coupled constraints that break standard MAPD solvers, reducing task completion to less than 10\% of each solver's own peak in some scenarios (see Section \ref{sec:experiment}).
Traditional MAPD algorithms \cite{maLifelongMultiagentPath2017,liuTaskPathPlanning2019} fail to account for human movement and dynamic constraints.
Addressing these limitations requires a novel approach that integrates real-time adaptability with congestion-aware planning to ensure effective coordination between AMRs and human workers.

To bridge this gap, we formalize the Multi-Agent Pickup-and-Delivery with Human Collaboration (MAPDwHC) problem, incorporating additional constraints such as \textbf{dynamic obstacles} and \textbf{real-time adaptability}.
We introduce \emph{CARVS} (\textbf{C}ongestion-\textbf{A}ware \textbf{R}andom-initialization, \textbf{V}alidation, and \textbf{S}wapping), a new algorithm designed to enhance task assignment and path planning in human-robot collaborative environments.
CARVS extends a baseline greedy assignment with three key improvements: random initialization, validation checks, and task swapping.
These enhancements mitigate congestion, adapt to human unpredictability, and improve real-time responsiveness.
Our theoretical analysis establishes the solvability of MAPDwHC, and extensive simulations demonstrate that CARVS significantly outperforms existing methods in terms of completed tasks.
This work advances the field of multi-agent systems and paves the way for optimizing human-robot collaboration in real-world logistics applications.

Our contributions are as follows:
\begin{itemize}
    \item We formalize the MAPDwHC problem and demonstrate that incorporating human collaboration heightens its computational difficulty relative to the traditional MAPD problem.
    \item We introduce \emph{CARVS}, an algorithm that integrates random initialization, validation checks, and task swapping to mitigate congestion, adapt to human-robot interactions, and maintain real-time responsiveness.
    \item Through extensive simulations, we show that \emph{CARVS} outperforms existing methods in terms of completed tasks, underscoring its practical potential in human-robot collaborative warehouse systems.
\end{itemize}

\section{Related Works}

\subsection{Multi-Agent Path Finding}

In Multi-Agent Pathfinding (MAPF), algorithms are broadly classified into optimal and suboptimal categories.
Optimal algorithms like CBS \cite{sharonConflictBasedSearchOptimal2012} and MA-CBS \cite{sharonMetaAgentConflictBasedSearch2012} aim to minimize the total cost but are computationally expensive, making them impractical for real-time applications with many agents.
As a result, suboptimal algorithms have become the focus of MAPF research, prioritizing runtime efficiency over strict optimality.

Suboptimal algorithms are typically divided into complete and incomplete methods.
Incomplete approaches, such as Rolling-Horizon Collision Resolution (RHCR) \cite{liLifelongMultiAgentPath2021}, simplify planning by resolving conflicts within limited timeframes but can suffer from deadlocks in constrained environments due to their non-exhaustive search.
On the other hand, complete suboptimal algorithms like LaCAM \cite{okumuraLaCAMSearchBasedAlgorithm2022} and LaCAM* \cite{okumuraImprovingLaCAMScalable2023} use dual-level search techniques that guarantee finding a feasible solution or proving its non-existence, though they can result in significant computational delays, particularly in complex scenarios like those encountered in MAPDwHC.

Recent studies \cite{lauMultiAgentPickupDelivery2022,damaniPRIMAL2PathfindingReinforcement2021} have applied multi-agent reinforcement learning (MARL) \cite{wangOrderMattersAgentbyagent2023} to MAPF problems, but these methods often struggle to meet the runtime requirements necessary for real-time human-robot collaboration in the MAPDwHC problem.
While MARL can improve adaptability and decision-making, it still faces challenges in balancing computational demands with the need for real-time solutions.

\subsection{Task Assignment}

Task assignment is often overlooked in traditional Multi-Agent Pathfinding (MAPF), resulting in inefficiencies as agents may be assigned suboptimal tasks that are difficult to coordinate.
To address this, several studies have integrated task assignment with path planning.
Nguyen et al. \cite{nguyenGeneralizedTargetAssignment2017} propose a generalized Combined Target Assignment and Path Finding (TAPF) model using Answer Set Programming, which handles scenarios with unequal numbers of agents and tasks, and incorporates complex requirements such as deadlines.

Honig et al. \cite{honigConflictbasedSearchOptimal2018} introduce Conflict-Based Search with Optimal Task Assignment (CBS-TA), integrating task assignment directly into pathfinding to optimize both simultaneously.
However, these integrated approaches often struggle with scalability in large or complex environments.
Okumura et al. \cite{okumuraSolvingSimultaneousTarget2022} address the unlabeled variant of MAPF, where any agent may be sent to any target.
Their algorithm TSWAP swaps targets between agents and uses a cached distance table to keep the planning cost of each time step low even at large agent counts.
All targets are given at the start rather than arriving over time.
While MAPD addresses dynamic environments with real-time task arrivals, TAPF is typically designed for static settings.

Ma et al. \cite{maLifelongMultiagentPath2017} extend the Lifelong MAPF framework to Multi-Agent Pickup and Delivery (MAPD), proposing algorithms like Token Passing (TP) and Token Passing with Task Swaps (TPTS) to improve flexibility and efficiency.
Building on this, Liu et al. \cite{liuTaskPathPlanning2019} develop MAPD strategies by framing task allocation as a Traveling Salesman Problem for each agent, while integrating deadlock avoidance mechanisms.

In a recent work, Xu et al. \cite{xuMultigoalMultiagentPickup2022} introduce algorithms for the Multi-Goal Multi-Agent Pickup and Delivery (MG-MAPD) problem, extending MAPD to handle tasks with multiple goal locations.
While their algorithms, LNS-PBS and LNS-wPBS, show significant scalability and efficiency improvements, particularly in large-scale settings, they focus primarily on offline environments and assume that tasks are known in advance.
In contrast, our work deals with real-time task arrivals and dynamic human-robot collaboration, which introduces additional complexities and requires more adaptive solutions.

Despite these advancements, challenges remain in managing dynamically increasing tasks and ensuring scalability in complex environments.
For instance, Chen et al. \cite{chenIntegratedTaskAssignment2021} present a unified approach combining task assignment and MAPF but may encounter efficiency issues in larger-scale settings.
Additionally, Krnjaic et al. \cite{krnjaicScalableMultiAgentReinforcement2022} apply Multi-Agent Reinforcement Learning to similar problems, though their methods do not address potential collisions, limiting their applicability in real-world scenarios.
Lodigiani et al. \cite{lodigianiRobustMultiagentPickup2023} study a related problem, MAPD with delays (MAPD-d), where agents may experience unexpected delays during execution. They propose $k$-TP and $p$-TP, two Token Passing variants offering deterministic and probabilistic robustness guarantees respectively, where $k$ bounds the delay an agent may absorb without triggering a replan. Their approach trades efficiency for robustness but does not support real-time replanning upon state changes.

\section{Online MAPD with Human Collaboration}

\subsection{Problem Definition}

We consider the environment as an undirected graph $G=(V,E)$.
At each discrete time step $t\in\{0,1,\dots ,T\}$, an agent $a_i\in A$ located at vertex $v\in V$ may either remain at $v$ or traverse an incident edge $(v,u)\in E$ to reach an adjacent vertex $u$.
One time step corresponds to a wall-clock duration of $\eta$ seconds in real time.

We are given a set of tasks $\mathcal{T}=\{\tau_1,\dots ,\tau_n\}$.
Task $\tau_j=(s_j,r_j)$, where $r_j$ is the release time, becomes available once the current time satisfies $t>r_j$.
Completion requires (i) an agent to travel to the pickup vertex $s_j\in V_P\subseteq V$, collect the package with human assistance, and then (ii) bring the package to some delivery vertex in $V_D\subseteq V$.
All tasks share a common set of packing stations $V_D$; upon pickup, an agent delivers to any available station, typically the nearest one.
Each agent can carry up to $\mathcal{K}$ packages simultaneously. We denote by $c_i$ the number of packages currently held by agent $a_i$.

Let $\mathcal{P}$ and $\mathcal{Q}$ be the disjoint sets of human pickers and packers, respectively, with $\mathcal{P}\cap\mathcal{Q}=\varnothing$, $|\mathcal{P}|>0$, and $|\mathcal{Q}|>0$.
Pickers travel to vertices in $V_P$ to load packages onto agents, whereas packers wait at vertices in $V_D$ to unload arriving packages.
An interaction begins when an agent arrives at a pickup or delivery vertex. The agent becomes stationary immediately and waits until the human worker arrives and completes the operation.
During an interaction the agent must remain stationary.
Let $\sigma(a_i,t)\in\{0,1\}$ be a binary indicator that is $1$ iff agent $a_i$ is interacting with a human at time $t$.
If an interaction begins at time step $t_s$ (i.e. $\sigma(a_i,t_s)=1$), it ends at $t_e=\min\{t>t_s\mid \sigma(a_i,t)=0\}$.
The interaction duration is $\delta=t_e-t_s$, which varies from one interaction to the next, with expected value $\mathbb{E}[\delta]$.

A path for agent $a_i$ is a sequence of vertices $p_i=(v_1,v_2,\dots ,v_k)$ specifying its location at every time step.
A collision at time $t$ occurs either when two agents occupy the same vertex (vertex collision $\langle a_i,a_j,v,t\rangle$) or traverse the same edge in opposite directions (edge collision $\langle a_i,a_j,e,t\rangle$).
Only agent-to-agent collisions are taken into account; human workers are not regarded as blocking.
A solution is the set of paths $P=\{p_1,\dots ,p_m\}$ for all agents and is \emph{feasible} when no collisions occur.
If, at a given time step, the planner fails to find a feasible continuation, all agents remain stationary for that step, whereas human workers proceed with their tasks.

Let $\phi(\tau,a_i,t)\in\{0,1,2,3\}$ denote the status of task $\tau$ with respect to agent $a_i$ at time $t$, where $0$ means unassigned, $1$ assigned, $2$ picked up, and $3$ delivered.
At time $t$, the tasks eligible for assignment form the set
$\mathcal{T}_t=\{\tau\in\mathcal{T}\mid t>r\text{ and }\forall a_i\in A:\phi(\tau,a_i,t)=0\}$.
Each task can be assigned to at most one agent at any given moment.
Conversely, each agent can hold at most $\mathcal{K}$ tasks at a time.

An instance of online MAPDwHC problem is the tuple
\[
\langle G, A, \mathcal{T}, \mathcal{P}, \mathcal{Q}, \mathcal{K}, \sigma(a_i,t), \eta, T\rangle.
\]
Its objective is to maximize the number of tasks completed within the horizon $T$, quantified as
\[
\widetilde{\mathcal{T}}
  =\Bigl|\{\tau_j\in\mathcal{T}\mid
          \exists a_i\in A,\,\exists t\le T:\phi(\tau_j,a_i,t)=3\}\Bigr|.
\]

\subsection{Human Interaction and Collaboration}

We provide an illustration of the MAPDwHC problem in Figure~\ref{fig:problem}.
The main difference between the MAPDwHC problem and the classical MAPD problem is that agents in MAPDwHC rely on human workers to complete pickup and delivery operations.
Consequently, while a planner computes paths, the environment keeps evolving: human workers continue to move and may arrive at pickup or delivery vertices, forcing the agent to wait at that vertex.

Since equipping every worker with positioning hardware adds significant cost, most warehouses do not track human positions in real time. Accordingly, the algorithm observes human--agent interactions only through the binary signal $\sigma(a_i,t)$.
Interaction durations $\delta$ are stochastic and depend on human-worker efficiency, so their end times are unknown in advance.
Whenever an interaction is active, the agent cannot move and thus becomes a \textbf{dynamic obstacle}.

\begin{property}[Dynamic Obstacle]\label{property:robot_movement}
Consider an agent $a_i$ occupying vertex $v_t^{a_i}$ at time $t$, its position at the next time step is
\[
v_{t+1}^{a_i} =
\begin{cases}
v_t^{a_i}, & \text{if } \sigma(a_i,t) = 1,\\
v,          & \text{otherwise},
\end{cases}
\]
where $v\in V$ is its planned next vertex whenever no interaction blocks movement.
\end{property}

To keep other agents collision-free, the planner can temporarily remove any interacting agent $a_i$ and its occupied vertex from consideration, yielding the reduced agent set $A'=A\setminus\{a_i\}$ and graph $G'=(V\setminus\{v_t^{a_i}\},E')$.

Another key aspect of MAPDwHC is that the interaction time $\delta$ is not fixed.
Because human workers do not follow fixed schedules, the start and duration of each interaction are stochastic; consequently, whenever an agent starts or finishes interacting with a human worker, or whenever the available task set $\mathcal{T}_t$ changes at time step $t$, the solution must be updated, requiring a replan for all agents.
If the algorithm fails to produce a feasible solution within the $\eta$-second window before the next time step, human workers continue their activities, but the agents cannot move; we say the algorithm \emph{skips} that time step.
Hence the planner must be \textbf{real-time responsive} to changes in the environment.

\begin{property}[Real-Time Responsiveness]\label{property:real_time_responsiveness}
The algorithm has a time budget of $\eta$ seconds at every time step to compute a feasible solution.
Failing to do so leaves all agents stationary during the next step, i.e., $v_{t+1}^{a_i}=v_t^{a_i}$.
\end{property}

Time steps discretize wall-clock time into intervals of $\eta$ seconds at which the system synchronizes plans to the agents, so $\eta$ is the minimum response time of the system to external changes.
An early solution is applied at the next step, and a late one is skipped while the planner keeps working on the same instance until a change in the environment triggers a replan.
While one could attempt to identify a safe subset of agents that can still move when the planner times out, this requires solving additional MAPF instances over exponentially many agent subsets, since the solver provides no information about which agents caused the failure.
The all-stop policy therefore offers a simple, safe fallback with no additional computation.

\subsection{Well-formed MAPDwHC Instance}

We first define the solvability of an online MAPDwHC instance as follows.

\begin{definition}[Solvable MAPDwHC Instance]\label{def:solvable_MAPDwHC}
An instance of online MAPDwHC is \emph{solvable} if there exists a complete algorithm that, at every time step~$t$, yields a feasible solution~$P$ satisfying Property~\ref{property:robot_movement}.
As the time budget $\eta\!\to\!\infty$ and the time step horizon $T\!\to\!\infty$, the algorithm eventually completes all tasks in the task set $\mathcal{T}$.
\end{definition}

We now outline several key properties to ensure the solvability of a MAPDwHC instance, and then define a well-formed MAPDwHC instance and prove it solvable.

\begin{property}[Map Connectivity]\label{property:map_connectivity}
    Let $\omega(G,v,u)$ denote the (finite) set of paths from vertex $v\in V$ to vertex $u\in V_P\cup V_D$.
    The map is \emph{connected} if
    \[
        \forall\,v\in V,\;\forall\,u\in V_P\cup V_D:\quad
        \bigl|\omega(G,v,u)\bigr|\;\ge\;1.
    \]
\end{property}

\begin{property}[Finite Interaction]\label{property:finite_interaction}
For any human-robot interaction, the interaction must be completed within a finite time frame, which is denoted as $\delta < \infty$.
\end{property}

\begin{property}[Sufficient Junctions]\label{property:map_layout}
For the graph $G=(V,E)$, define the set of junction vertices as
\[
  V_F \;=\; \bigl\{\, v \in V \mid v \notin V_P \cup V_D,\; \deg(v)\ge 3 \bigr\}.
\]
The graph must satisfy $|V_F| \;\ge\; |A|.$
\end{property}

\begin{definition}[Well-formed MAPDwHC Instance]
A well-formed MAPDwHC instance is an online MAPDwHC instance that satisfies Properties \ref{property:robot_movement}, \ref{property:map_connectivity}, \ref{property:finite_interaction} and \ref{property:map_layout}.
\end{definition}

Note that Property~\ref{property:real_time_responsiveness} (Real-Time Responsiveness) is excluded from the well-formedness definition because it does not affect solvability: even if the algorithm exceeds the time budget~$\eta$, keeping all agents at their current positions is always a feasible solution for that time step. Well-formedness concerns structural solvability, not real-time performance.

To transform the current task allocation into a solvable MAPF instance, the allocation algorithm must satisfy the following assumption:

\begin{assumption}[Unique Assignment]\label{assump:unique-task}
    At every time step, no two agents are assigned the same vertex as their goal.
\end{assumption}

Note that multiple tasks in $\mathcal{T}$ may share the same pickup vertex, but at most one of them can be assigned at any given time step.

\begin{theorem}\label{th:1}
    With the task assignment that satisfies Assumption~\ref{assump:unique-task}, any well-formed MAPDwHC instance is solvable by a complete MAPF algorithm.
\end{theorem}

\begin{proof}
    At each time step $t$, since the number of available agents is $|A|$, we select a subset $\mathcal{T}'\subseteq\mathcal{T}_t$ of at most $|A|$ tasks, each assigned to a distinct agent satisfying Assumption~\ref{assump:unique-task}. When $|\mathcal{T}_t|>|A|$, the remaining tasks wait for subsequent time steps.
    This assignment yields a valid MAPF instance.
    If all agents are free to move (i.e. $\sum_{a_i \in A} \sigma(a_i,t)=0$), Properties~\ref{property:map_connectivity} and \ref{property:map_layout} guarantee that any complete MAPF algorithm can solve this instance.
    When some agents are blocked (i.e. $\sum_{a_i \in A} \sigma(a_i,t)>0$), as described in Property~\ref{property:robot_movement}, we can create the reduced agent set $A'=A\setminus\{a_i\}$ and graph $G'=(V\setminus\{v_t^{a_i}\},E')$.
    If this removal break Properties~\ref{property:map_connectivity} or \ref{property:map_layout} and create a deadlock, Property~\ref{property:finite_interaction} ensures that Properties~\ref{property:map_connectivity} or \ref{property:map_layout} will be restored after finite time steps.  Once the deadlock is resolved and all agents are again free to move, the MAPF instance becomes solvable.
    Hence, every well-formed MAPDwHC instance is solvable.
    A detailed proof is provided in Appendix~\ref{app:extended_proof}.
\end{proof}

\subsection{Relation to MAPD}

\begin{theorem}
    \label{th:mapd}
    A well-formed MAPDwHC instance can be converted into a sequence of MAPD instances over time.
\end{theorem}

\begin{proof}
    If there is no interaction at time step $t$, it reduces directly to a valid MAPD instance with the task set $\mathcal{T}_t$, the graph $G$ and the agents $A$.
    During an interaction, if an agent $a_i$ becomes stationary at vertex $v_t^{a_i}$, we obtain a MAPD instance with the reduced agent set $A' = A \setminus \{a_i\}$ and the reduced graph $G' = (V \setminus \{v_t^{a_i}\}, E')$. The instance remains solvable as long as Properties~\ref{property:map_connectivity} and \ref{property:map_layout} hold.
\end{proof}

We leverage Theorem~\ref{th:mapd} to adapt MAPD algorithms for evaluation within MAPDwHC scenarios and enable comparisons with the proposed algorithms.
However, during human-robot interactions, preserving previous task assignments may render MAPD instances unsolvable on the reduced graph $G'$.
For example, removing an occupied vertex may split $G'$ into disconnected components, causing tasks that span multiple components to become infeasible to complete.
Proving unreachability often requires more computation than finding a feasible plan; for instance, A*-based algorithms \cite{Hart1968Astar} may need to explore the entire search space before concluding that no path exists.
MAPD algorithms typically do not account for this.
Our approach explicitly incorporates it and delivers an improved solution.

\section{Method}

\subsection{Greedy Assignment}

Since both MAPD and MAPF are NP-hard \cite{maMultiAgentPathFinding2016}, finding optimal solutions within reasonable time is challenging.
We simplify the problem by separating it into two levels: task assignment and pathfinding.
At the high level, tasks are assigned to agents, and at the low level, existing MAPF algorithms solve the pathfinding problem.

A straightforward yet suboptimal approach to solve the task assignment in MAPDwHC is the First-Come, First-Served (FCFS) strategy, where tasks are processed strictly in arrival order.
While simple to implement, FCFS often causes unnecessary back-and-forth travel.
A more effective alternative is the Greedy Assignment (GA) algorithm.
Let $a_i.g$ denote the current goal vertex of agent $a_i$, and let $h(\cdot,\cdot)$ be a heuristic function that estimates the travel cost between two vertices (e.g., Manhattan distance).
For each unassigned agent $a_i$, if its current load $c_i < \mathcal{K}$, GA selects the nearest available task as its goal, i.e., $a_i.g \gets \min_{s_i \in \mathcal{T}_t} h(a_i.v, s_i)$, where $a_i.v$ is the agent's current vertex; if $c_i = \mathcal{K}$ (capacity full), the agent is directed to the nearest delivery point, i.e., $a_i.g \gets \min_{d_i \in V_D} h(a_i.v, d_i)$.
To respect Property~\ref{property:robot_movement}, any agent currently interacting with a human has its position removed from the graph.
The resulting set of start-goal pairs is then passed to a complete MAPF algorithm to compute collision-free paths.
Note that FCFS is a special case of GA where the heuristic $h$ assigns tasks in arrival order.

\begin{theorem}\label{th:ga}
    A well-formed MAPDwHC instance can be solved by the greedy assignment algorithm in combination with any complete MAPF algorithm, together with a heuristic function $h$ that ensures no agents have overlapping goals.
\end{theorem}

\begin{proof}
    We can design a heuristic function $h$ that selects a goal based primarily on the Manhattan distance, while ensuring that no agents have overlapping goals.
    Given Property~\ref{property:map_connectivity}, the selected goals form a MAPF instance that can be solved by any complete MAPF algorithm.
    Although Property~\ref{property:map_connectivity} may not always be satisfied when some agents are interacting with humans, Theorem~\ref{th:1} proves that the MAPF instance will become solvable within a finite number of time steps.
    By applying Theorem~\ref{th:1}, we can conclude that Theorem~\ref{th:ga} also holds.
\end{proof}

Complete MAPF algorithms always return either a valid plan or a proof that no solution exists.
On unsolvable instances, however, demonstrating unsolvability still forces the search to (eventually) explore the entire search space, so the runtime can grow prohibitively large—even when advanced pruning or ordering heuristics are used.

\subsection{Solution Refinement}

We propose \emph{CARVS} (\textbf{C}ongestion-\textbf{A}ware \textbf{R}andom-initialization, \textbf{V}alidation, and \textbf{S}wapping), an approach that enhances a baseline greedy assignment with three key improvements, as outlined in Algorithm \ref{alg:greedy_refine}.
CARVS is designed to optimize task allocation while keeping the underlying logic lightweight, thereby ensuring efficient performance in real-world deployments.
The key refinement steps are highlighted in blue.
The refinement algorithm is composed of two main parts, with the three improvements illustrated in Figure~\ref{fig:methods}.
Together, these enhancements make CARVS more robust than purely greedy methods in large-scale warehouses.

CARVS adopts a complete MAPF solver as its back-end.
Leveraging this completeness, CARVS is guaranteed to solve any well-formed MAPDwHC instance (a formal proof is given later).
The method's three enhancements are specifically crafted to prevent the excessive runtimes that complete MAPF solvers can incur when temporary blockages arise, enabling CARVS to pair seamlessly with any future, more advanced complete MAPF algorithm.
Although the solvers used in this paper are centralized, any centralized or distributed MAPF solver that retains the completeness property could serve equally well as CARVS's underlying planner.

\begin{algorithm}[htbp]
\footnotesize
\caption{CARVS Algorithm}
\label{alg:greedy_refine}
\begin{algorithmic}[1]
\REQUIRE Instance $\langle G,A,\mathcal{T},\mathcal{P},\mathcal{Q},\mathcal{K},\sigma(a_i,t),\eta, T \rangle$, PF algorithm $\omega$, MAPF algorithm $\Omega$, Heuristic function $h$
\ENSURE Solution $P$
\STATE \textit{// Greedy assignment}
\FOR{each $a_i \in A$}
\label{algo:greedy_assign_start}
    \IF{$\exists a_i.g$}
        \STATE \textbf{continue}
    \ENDIF
    \IF{$c_i < \mathcal{K}$}
        \IF{$c_i = 0$}
        \label{algo:random_first_start}
            \STATE $a_i.g \gets \text{rand}_{s_i \in \mathcal{T}_t}(s_i)$ \COMMENT{random initialization}
            \label{algo:random_first}
        \ELSE
        \label{algo:random_first_end}
            \STATE $a_i.g \gets \min_{s_i \in \mathcal{T}_t} h(a_i.v, s_i)$
        \ENDIF
    \ENDIF
    \IF{$c_i = \mathcal{K}$}
        \STATE $a_i.g \gets \min_{d_i \in V_D} h(a_i.v, d_i)$
    \ENDIF
\ENDFOR
\label{algo:greedy_assign_end}
\STATE Initialize set $M$; Duplicate $G$ as $G'$
\FOR{each $a_i \in A$}
    \IF{$\nexists a_i.g$}
        \STATE \textbf{continue}
    \ENDIF
    \IF{$\sigma(a_i,t)=1$}
        \STATE remove $a_i.v$ from $G'$
    \ENDIF
\ENDFOR
\STATE \textit{// Swap goals}
\FOR{each $a_i, a_j \in A$ where $i \neq j$}
\label{algo:swap_goal_start}
    \IF{$\phi(\tau_i,a_i,t) \neq 1$ \textbf{or} $\phi(\tau_j,a_j,t) \neq 1$}
        \STATE \textbf{continue}
    \ENDIF
    \STATE $p_{ii} \gets \omega(G', a_i.v, a_i.g)$; $p_{jj} \gets \omega(G', a_j.v, a_j.g)$
    \IF{$a_i.v \notin p_{jj}$ \textbf{or} $a_j.v \notin p_{ii}$}
        \STATE \textbf{continue}
    \ENDIF
    \STATE $p_{ij} \gets \omega(G', a_i.v, a_j.g)$; $p_{ji} \gets \omega(G', a_j.v, a_i.g)$
    \IF{$a_i.v \notin p_{ji}$ \textbf{or} $a_j.v \notin p_{ij}$}
        \STATE swap $a_i.g$ and $a_j.g$ \COMMENT{swap goals}
    \ENDIF
\ENDFOR
\label{algo:swap_goal_end}
\STATE \textit{// Plan path}
\FOR{each $a_i \in A$}
    \IF{$\sigma(a_i,t)=1$}
        \STATE \textbf{continue}
    \ENDIF
    \IF{$\exists a_i.g$ \textbf{and} $\omega(G', a_i.v, a_i.g) \neq \emptyset$}
        \STATE $M.\text{add}(\langle a_i.v, a_i.g \rangle)$
        \STATE \textbf{continue}
    \ENDIF
    \STATE $w_i \gets \min_{w_i \in V_W} h(a_i, w_i)$ \COMMENT{find waiting point}
    \STATE $M.\text{add}(\langle a_i.v, w_i \rangle)$
\ENDFOR
\STATE $P \gets \Omega(G', M)$
\RETURN{$P$}
\end{algorithmic}
\end{algorithm}

\subsubsection{Validate Assignment}
The first part of the improvements focuses on refining the current assignment.
A key idea is task swapping, which is only allowed during pickup tasks before items are collected (i.e. $\phi(\tau_i,a_i,t) = \phi(\tau_j,a_j,t) = 1$).
When two agents' shortest paths intersect at each other's current positions and both are heading to pickup points, the algorithm can swap their tasks.
Since agents can not swap physical items with one another, delivery tasks will not be swapped between them.
Another important idea is to validate whether an agent can reach its goal using a path planning algorithm $\omega$ (e.g. A*).
If the agent cannot reach its goal, a waiting point \( w_i \in V_W \) will replace its real goal.
The set of waiting points $V_W$ is defined as the set of non-cut vertices that are not pickup or delivery locations:
\[
  V_W \;=\; \bigl\{\, v \in V \mid v \notin V_P \cup V_D,\; G \setminus \{v\} \text{ remains connected} \bigr\}.
\]

The heuristic function $h$ in the algorithm can select a waiting point near the agent's pickup or delivery point based on the Manhattan distance, while ensuring that no agents have overlapping goals.
When an agent's goal is temporarily unreachable, redirecting it to a nearby waiting point alleviates congestion in key areas.
Because every $w \in V_W$ is a non-cut vertex of $G$, stationing a single agent at $w$ preserves graph connectivity.
Using multiple waiting points simultaneously can still help alleviate congestion.
However, graph partitions caused by human-robot interactions cannot be resolved by goal reassignment; Theorem~\ref{th:1} guarantees restoration of solvability once such interactions conclude.

\subsubsection{Random Initialization}

Random initialization modifies the initial assignment for a free agent by randomly assigning its first goal from the available tasks $\mathcal{T}_t$.
This adjustment helps alleviate congestion, particularly in areas with uneven distributions of pickup points or high task density.
Without this improvement, if a robot consistently selects the nearest goal after unloading, it may end up handling only tasks near the drop-off zone, especially in regions where new tasks frequently appear.
While the greedy assignment performs well on open maps, it can lead to problems in environments with narrow passages.
In such cases, corridors may be blocked for extended periods, increasing robot density in certain areas and causing congestion.
By selecting the first goal randomly and then choosing the nearest goal for subsequent tasks, CARVS effectively disperses the robots across the map, reducing bottlenecks and ensuring smoother overall operation.

\begin{theorem}\label{th:CARVS}
    A well-formed MAPDwHC instance can be solved by CARVS in combination with any complete MAPF algorithm, with a heuristic function $h$ that ensures no agents have overlapping goals.
\end{theorem}

\begin{proof}
   We can design a heuristic function \( h \) that selects a goal based primarily on the Manhattan distance, ensuring that no agents have overlapping goals.
   Given Property~\ref{property:map_connectivity}, the selected goals form a MAPF instance that can be solved by any complete MAPF algorithm.
   Although Property~\ref{property:map_connectivity} may not always be satisfied when some agents interact with humans, Theorem~\ref{th:1} guarantees that the MAPF instance will become solvable within a finite number of time steps.
   Therefore, a valid MAPF instance can be constructed from the current task assignment.
   Swapping goals does not invalidate the MAPF instance, and if no suitable waiting point is found, retaining the previous assignment is acceptable.
   By following Theorem~\ref{th:1}, we can conclude that Theorem~\ref{th:CARVS} holds.
\end{proof}

\begin{figure*}[!t]
    \centering
    \includegraphics[width=\textwidth]{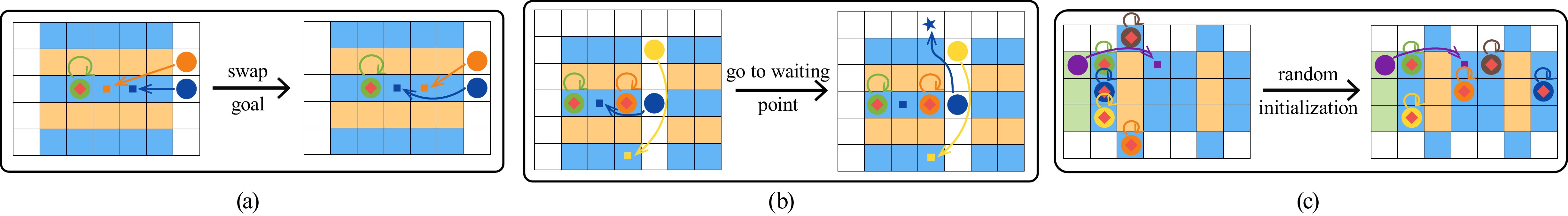}
    \caption{A diagram illustrating optimization methods: the picking area is blue, the packing area green, and shelves are orange. Circles represent agents, diamonds represent humans, and squares represent goal points, with arrows indicating robot targets. (a) Goal swapping reduces delays from intersecting paths. (b) Stars mark wait points for blocked agents. (c) Greedy assignment leads to congestion near the packing area, mitigated by random initial goal assignment for idle agents.}
    \label{fig:methods}
    \Description{A diagram illustrating optimization methods: the picking area is blue, the packing area green, and shelves are orange. Circles represent agents, diamonds represent humans, and squares represent goal points, with arrows indicating robot targets. (a) Goal swapping reduces delays from intersecting paths. (b) Stars mark wait points for blocked agents. (c) Greedy assignment leads to congestion near the packing area, mitigated by random initial goal assignment for idle agents.}
\end{figure*}

\begin{figure}[!t]
  \centering
  \includegraphics[width=\columnwidth]{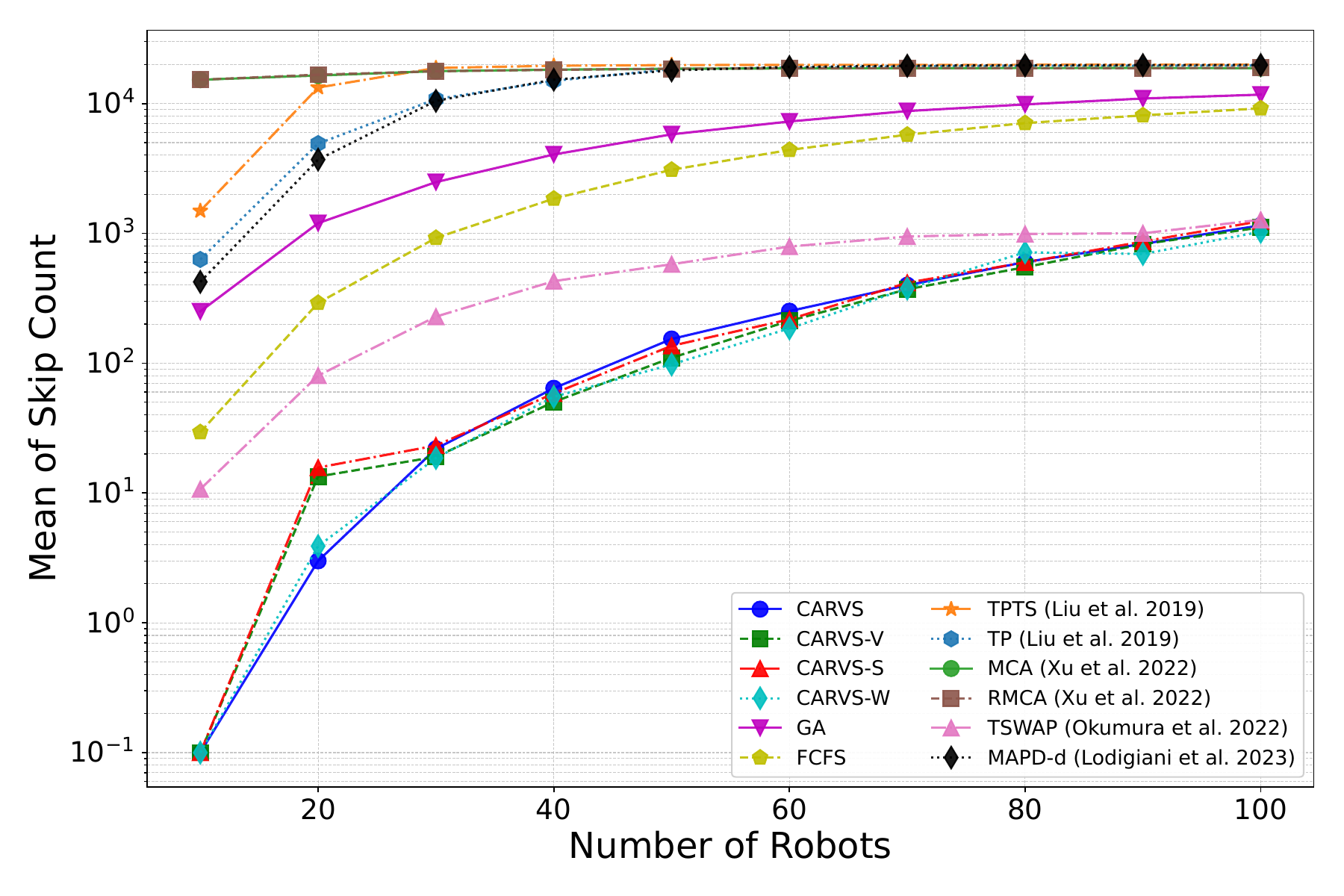}
  \caption{Mean skip count on the \texttt{warehouse\_small} map for 10 to 100 robots. A skip occurs when the planner fails to return a solution within the time limit. CARVS variants substantially reduce skips as the robot count grows.}
  \label{fig:not_found}
  \Description{Mean skip count on the \texttt{warehouse\_small} map for 10 to 100 robots. A skip occurs when the planner fails to return a solution within the time limit. CARVS variants substantially reduce skips as the robot count grows.}
\end{figure}

\section{Experiment}
\label{sec:experiment}

\begin{figure*}[!t]
  \centering
  \includegraphics[width=\textwidth]{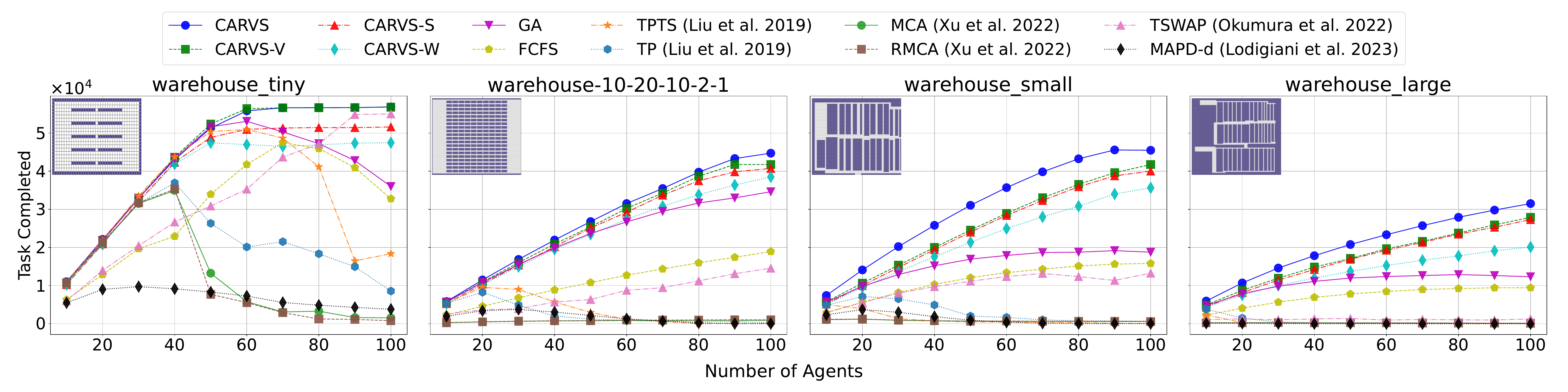}
  \caption{The performance metrics of various algorithms across multiple warehouse configurations, which show the mean number of completed tasks across 10 random seeds as a function of agent count, ranging from 10 to 100. CARVS scales effectively with the number of agents, outperforming other algorithms.}
  \label{fig:WMAPF-metrics}
  \Description{The performance metrics of various algorithms across multiple warehouse configurations, which show the mean number of completed tasks across 10 random seeds as a function of agent count, ranging from 10 to 100. CARVS scales effectively with the number of agents, outperforming other algorithms.}
\end{figure*}

Following Theorem~\ref{th:mapd}, we can compare CARVS against traditional MAPD algorithms.
The experiments aim to demonstrate the effectiveness of CARVS combined with MAPF algorithms in solving the MAPDwHC problem, with a focus on potential real-world applications.
We use a simulation that aligns human activities with real-world scenarios for comparisons.
The $\eta$ is set to 1 second.
Failing to produce a solution within this time limit results in a skipped time step, where all agents remain stationary until the next time step.
Each time step, a robot moves 0.75 meters (one edge long), while humans, moving at 1.5 m/s, are set to twice the robot's speed.
In a real-world deployment the signal $\sigma(a_i,t)$ is triggered directly by the human worker.
In the simulation, human workers are following an A* path and $\sigma(a_i,t)$ is generated from the predicted trajectory.
In the simulation, the interaction duration follows a degenerate distribution with $\Pr(\delta=10)=1$, and hence $\mathbb{E}[\delta]=10$ time steps (10 seconds).
Performance is evaluated based on total tasks completed.
We test several recent MAPF algorithms with CARVS, including EECBS \cite{liEECBSBoundedSuboptimalSearch2021}, RHCR \cite{liLifelongMultiAgentPath2021}, CBSH2-RTC \cite{liPairwiseSymmetryReasoning2021}, LaCAM \cite{okumuraLaCAMSearchBasedAlgorithm2022}, LaCAM2 \cite{okumuraImprovingLaCAMScalable2023}, and MAPF-LNS2 \cite{liMAPFLNS2FastRepairing2022}.
Not all MAPF algorithms perform well on complex maps with many agents, so we select LaCAM2 for its best performance.
A detailed comparison of the MAPF algorithms is provided in Appendix~\ref{app:mapf}.

Most of the MAPF algorithms are implemented in C++, and we have added a wrapper to make them compatible with our simulator and the MAPDwHC problem.
The wrapper generates a map file and a MAPF instance for the algorithm.
After the program completes, the wrapper reads and parses the output solution for the simulation.
For the MAPD algorithms, we include results from TP/TPTS \cite{liuTaskPathPlanning2019}, MCA/RMCA \cite{xuMultigoalMultiagentPickup2022}, TSWAP \cite{okumuraSolvingSimultaneousTarget2022}, and MAPD-d \cite{lodigianiRobustMultiagentPickup2023}. We run MAPD-d as $k$-TP with $k=0$, because our setting provides no reliable way to fix $k$ in advance and a larger $k$ degrades performance sharply. Appendix~\ref{appendix:tpts_adaptation} explains this choice.
These algorithms were adapted for our simulation environment, with detailed modifications provided in Appendix~\ref{appendix:tpts_adaptation}.

The simulator and task assignment algorithms were developed in Python and run on a server with an AMD EPYC 7H12 64-Core Processor and 128GB RAM.
Each configuration was tested with 10 random seeds affecting agent positions, task allocation, and MAPF algorithms, with results averaged across all seeds.

\subsection{CARVS Performance}

For ablation study, we use CARVS-W to denote the CARVS variant that includes only waiting point improvement and CARVS-S to denote the CARVS variant that includes only swapping tasks.
CARVS-V is the CARVS variant that includes the Validate Assignment improvement but omits Random Initialization.
Note that removing all enhancements from CARVS reduces it to the greedy algorithm (GA).
Our evaluation spans four distinct map layouts, as depicted in Figure~\ref{fig:WMAPF-metrics}, with performance measured using total completed tasks $\widetilde{\mathcal{T}}$.
Each algorithm is tested over 20,000 time steps.

Algorithm scalability is a critical factor in performance assessment.
Figure \ref{fig:WMAPF-metrics} illustrates how task completion varies across different map layouts and agent counts.
As map congestion increases, most algorithms exhibit a decline in completed tasks.
However, CARVS demonstrates remarkable robustness, with its performance advantage becoming more pronounced as map difficulty and agent count increase.

By contrast, TP/TPTS and MCA/RMCA degrade sharply even in the \texttt{warehouse\_tiny} benchmark: once the agent count surpasses 60, their task-completion curves plunge, and at 100 agents they process no more than 10\% of their own peak workload.
In the larger warehouse settings their performance collapses further, completing virtually no tasks at all.
While FCFS completes fewer tasks than GA and CARVS, it consistently outperforms TP/TPTS and MCA/RMCA on larger maps.
TP and TPTS evaluate every agent-task pair; when an agent is free, they assign the task and generate a plan using prioritized planning.
Unfortunately, prioritized planning can still fail on some MAPF instances.
MCA and RMCA, in turn, build a separate task heap for each task, calling the path-planning routine in the process.
Because the task set changes continuously, these heaps must be rebuilt at every timestep, which is time-consuming.
That's why both TP/TPTS and MCA/RMCA scale poorly.

The two remaining MAPD baselines also perform poorly on complex maps, although only one is limited by planning cost.
TSWAP rarely exceeds the time limit because each step uses a greedy move guided by cached distances and local conflict rules rather than a joint search.
Its task completion increases on the three smaller layouts but remains nearly constant on the largest one.
TSWAP resolves conflicts by swapping targets rather than rerouting agents.
However, TSWAP is designed for unlabeled targets.
MAPD distinguishes pickup and drop-off targets.
The adapter can only restore the intended assignment in a later planning attempt after a target swap overrides it.
A delivery may therefore require multiple planning attempts even though each attempt is fast.
With 100 robots on \texttt{warehouse\_large}, less than 2\% of successful TSWAP paths for dropping robots ended at the assigned drop-off target.
This explains the low efficiency of drop-off tasks.
MAPD-d fails in the opposite way. Its replanning cost grows with the agent count until it exceeds the time budget at almost every step, and task completion drops to 2 at 100 agents.

Although GA outperforms FCFS, it struggles on complex maps such as \texttt{warehouse\_small} and \texttt{warehouse\_large}.
The non-uniform task distribution leads to agent crowding in certain areas, complicating path planning.
Additionally, long tunnels and dynamic obstacles further reduce GA's efficiency.
Due to the lack of full integration between high-level task assignment and low-level path planning, agents often remain idle when path planning fails, reducing efficiency.
In contrast, CARVS-V and CARVS validate goal reachability beforehand, improving the chances of solving the MAPF subproblem in time.
The validation improvements increase task completion by at least 2x on \texttt{warehouse\_large} with 100 agents.

CARVS further refines CARVS-V by mitigating agent congestion near drop-off zones.
GA and CARVS-V tend to select pickup locations close to drop-off points, leading to bottlenecks.
CARVS circumvents this issue by randomly selecting the first pickup locations, ensuring a more even agent distribution.
This enhancement increases task completion by at least 1.1x on \texttt{warehouse\_large} with 100 agents.

Although the performance gap among CARVS variants appears modest in Figure~\ref{fig:WMAPF-metrics}, this is largely because the y-axis spans the full range needed to display all algorithms, compressing the differences among variants that share a similar performance tier.
More fundamentally, planning timeouts are a major bottleneck for search-based MAPDwHC solvers.
When the MAPF solver exceeds the time budget $\eta$, all agents stall for that step.
Every CARVS variant incorporates mechanisms that reduce such timeouts and avoids the timeout-driven losses seen in several baselines.
Because these mechanisms are shared, the performance differences among the variants remain relatively small.

Finally, Figure~\ref{fig:not_found} illustrates the impact of different algorithms on skips as the number of agents increases on the \texttt{warehouse\_small} map.
A skip occurs when an algorithm exceeds the time limit.
Switching from GA to CARVS significantly reduces timeouts, by about 90.61\% at the highest robot count, demonstrating CARVS's improved efficiency in complex multi-agent scenarios, enhancing scalability and reliability.
Although TSWAP also keeps its skip count low, its target-swapping mechanism is not suitable for MAPD tasks.

Our largest experiments use 100 robots, far fewer than typical MAPF benchmarks, since the planner in MAPDwHC has only $\eta = 1\,\text{s}$ at every step while a MAPF solver usually gets a long limit for one instance, for example ten seconds in LaCAM*~\cite{okumuraImprovingLaCAMScalable2023}.
Under this budget the MAPD baselines already complete almost no tasks on the complex maps at 100 robots, while CARVS still does, so this setting is sufficient to show the difference.

\subsection{Additional Experiment Settings}

Apart from the choice of MAPF algorithms, the performance can vary based on three factors: map layout, task distribution, and the human-robot ratio.
Below, we outline the settings for each factor.

\subsubsection{Map Layout}
The layout of the map significantly impacts algorithm performance.
We scanned two real-world warehouse layouts, converted them into 2D grid maps, and formulated them as MAPF benchmarks \cite{sternMultiAgentPathfindingDefinitions2019}.
These maps were used to evaluate algorithm performance.
We estimate map complexity using five structural metrics: total vertex count, obstacle density, number of tunnels (our term for narrow aisles), average tunnel length, and the proportion of vertices with degree~$\ge 3$.
Table~\ref{table:map-layout} reports these statistics on the maps we used in the experiments.

\begin{table}[!t]
  \centering
  \small
  \setlength{\tabcolsep}{3pt}
  \caption{Metrics on Map Layout}
  \label{table:map-layout}
  \begin{tabular}{lccccc}
    \toprule
                           & Vert \# & Obst \% & Tun \# & Avg. L & Deg-3 \% \\ \midrule
    warehouse\_tiny        & 851      & 25.38     & 5         & 3        & 98.58      \\
    warehouse-10-20-10-2-1 & \textbf{10143}    & 43.81     & 38        & 8.2     & 56.76      \\
    warehouse\_small       & 4785     & 66.54     & 19        & \textbf{14.7}    & 82.82      \\
    warehouse\_large       & 8611     & \textbf{79.4}      & 78        & 13.2    & \textbf{43.12}      \\ \bottomrule
  \end{tabular}
\end{table}

The \texttt{warehouse\_tiny} map originates from the TP/TPTS study, whereas \texttt{warehouse\_small} and \texttt{warehouse\_large} are abstractions of real-world warehouse layouts.
The latter two feature longer tunnels and a higher obstacle density, making conflict resolution substantially harder.
When an agent engages in human interaction inside a tunnel it blocks the entire aisle.
Any robots within a tunnel must rely on nearby high-degree vertices (degree~$\ge 3$) to swap positions and resolve edge conflicts.
Consequently, the difficulty of solving a MAPF instance on a map is largely governed by three layout characteristics: the number of tunnels, their average length, and the number of degree-3 (or higher) vertices.

\subsubsection{Task Distribution}
Temporal and spatial task distributions significantly impact algorithm performance.
We collect real-world task sets from a warehouse management system, which records tasks of SKUs over time.
The original task sets contain a limited number of tasks and can be easily solved by CARVS and its variants, so they cannot distinguish the performance of different algorithms.
To enable a more comprehensive evaluation, we repeat each task from the original task sets 10 times to create larger task sets.
Tasks are then continuously added to the task pool over time following the original distribution.
The resulting spatial distribution is non-uniform, with certain areas having higher task densities than others.

\subsubsection{Human-Robot Ratio}
Human interaction increases the difficulty of the MAPDwHC problem, and algorithm performance can be affected by the ratio of humans to robots.
In our simulation, humans avoid collisions with robots and follow A* paths to reach their targets.
It is important to note that this behavior is confined to the simulation; in real-world applications, humans manually update their state changes in the system via a tablet.
Unlike traditional MAPD scenarios, boosting processing capacity requires scaling both the robots and the human workforce.
In our experiments, we fixed the number of human operators at 50, which nearly suffices for 100 robots.
Too few human operators reduce picking efficiency and lengthen the robots' idle periods, whereas too many operators increase their own waiting times.
In practice, we aim to keep human operators fully engaged, which usually means employing fewer humans than robots.
Because workforce size materially influences overall system performance, we hold the number of human operators constant across all experiments to allow a fair comparison of competing algorithms.
Human efficiency also impacts performance.
Human collaboration heightens the problem's difficulty relative to the traditional MAPD formulation, as discussed further in Appendix~\ref{app:human_eff}.

\section{Conclusion}

This paper formalizes the Multi-Agent Pickup and Delivery with Human Collaboration (MAPDwHC) problem, extending the traditional MAPD framework by integrating human agents and addressing the additional challenges introduced by human-robot interactions.
Although a MAPDwHC instance reduces to a sequence of MAPD problems when no interaction occurs, our empirical results confirm that MAPDwHC poses significantly greater computational challenges than MAPD in realistic warehouse settings.

To tackle these challenges, we propose CARVS, an approach that optimizes task assignment and path planning while accounting for congestion, human-robot interactions, and real-time adaptability.
We formally prove that CARVS can solve a well-formed MAPDwHC instance.
Extensive simulations demonstrate that our methods outperform existing approaches in terms of completed tasks, highlighting their potential for real-world applications.
This work advances multi-agent systems research and lays the foundation for future optimization of human-robot collaboration.
Beyond reducing planning timeouts, our results show that congestion-aware task assignment is also important.
Fast planning alone does not guarantee high task completion.
CARVS deliberately avoids optimizing task assignment and path planning as one model.
The joint search space of goal selection and collision-free paths is too large for the time budget of one step, and methods that couple the two, such as TP/TPTS and MCA/RMCA, break down on the complex maps.
Keeping the two levels separate is what lets the system meet its real-time requirement.
Future work could focus on incorporating order distribution as a parameter into the task assignment algorithm to improve robot efficiency.
Developing efficient partial-movement strategies when the planner times out is a promising direction for future work.

\begin{acks}
The authors would like to thank the anonymous reviewers for their valuable comments and suggestions that helped improve the quality of this paper.
\end{acks}

\bibliographystyle{ACM-Reference-Format}
\bibliography{MyLib}


\begin{thebibliography}{30}


\ifx \showCODEN    \undefined \def \showCODEN     #1{\unskip}     \fi
\ifx \showISBNx    \undefined \def \showISBNx     #1{\unskip}     \fi
\ifx \showISBNxiii \undefined \def \showISBNxiii  #1{\unskip}     \fi
\ifx \showISSN     \undefined \def \showISSN      #1{\unskip}     \fi
\ifx \showLCCN     \undefined \def \showLCCN      #1{\unskip}     \fi
\ifx \shownote     \undefined \def \shownote      #1{#1}          \fi
\ifx \showarticletitle \undefined \def \showarticletitle #1{#1}   \fi
\ifx \showURL      \undefined \def \showURL       {\relax}        \fi
\providecommand\bibfield[2]{#2}
\providecommand\bibinfo[2]{#2}
\providecommand\natexlab[1]{#1}
\providecommand\showeprint[2][]{arXiv:#2}
\makeatletter
\@ifundefined{NAT@parse@date}{}{\let\NAT@parse@date@orig\NAT@parse@date}
\@ifundefined{NAT@parse@date}{}{\def\NAT@parse@date#1#2#3#4#5#6@@{\NAT@parse@date@orig#1#2#3#4#5#6@@\def\NAT@tempyear{0000}\def\NAT@tempexlab{{?}}\ifx\NAT@year\NAT@tempyear\ifx\NAT@exlab\NAT@tempexlab\def\NAT@date{[n.\,d.]}\else\edef\NAT@date{[n.\,d.]\NAT@exlab}\fi\fi}}
\makeatother

\bibitem[Chen et~al\mbox{.}(2021)]%
        {chenIntegratedTaskAssignment2021}
\bibfield{author}{\bibinfo{person}{Zhe Chen}, \bibinfo{person}{Javier
  Alonso-Mora}, \bibinfo{person}{Xiaoshan Bai}, \bibinfo{person}{Daniel~D.
  Harabor}, {and} \bibinfo{person}{Peter~J. Stuckey}.}
  \bibinfo{year}{2021}\natexlab{}.
\newblock \showarticletitle{Integrated {Task} {Assignment} and {Path}
  {Planning} for {Capacitated} {Multi}-{Agent} {Pickup} and {Delivery}}.
\newblock \bibinfo{journal}{\emph{IEEE Robotics and Automation Letters}}
  \bibinfo{volume}{6}, \bibinfo{number}{3} (\bibinfo{date}{July}
  \bibinfo{year}{2021}), \bibinfo{pages}{5816--5823}.
\newblock
\showISSN{2377-3766}
\href{https://doi.org/10.1109/LRA.2021.3074883}{doi:\nolinkurl{10.1109/LRA.2021.3074883}}


\bibitem[Damani et~al\mbox{.}(2021)]%
        {damaniPRIMAL2PathfindingReinforcement2021}
\bibfield{author}{\bibinfo{person}{Mehul Damani}, \bibinfo{person}{Zhiyao Luo},
  \bibinfo{person}{Emerson Wenzel}, {and} \bibinfo{person}{Guillaume
  Sartoretti}.} \bibinfo{year}{2021}\natexlab{}.
\newblock \bibinfo{title}{{PRIMAL2}: {Pathfinding} via {Reinforcement} and
  {Imitation} {Multi}-{Agent} {Learning} -- {Lifelong}}.
\newblock
\href{https://doi.org/10.1109/LRA.2021.3062803}{doi:\nolinkurl{10.1109/LRA.2021.3062803}}


\bibitem[Duinkerken and Lodewijks(2015)]%
        {duinkerkenRoutingAGVsAutomated2015}
\bibfield{author}{\bibinfo{person}{Mark~B. Duinkerken} {and}
  \bibinfo{person}{Gabriel Lodewijks}.} \bibinfo{year}{2015}\natexlab{}.
\newblock \showarticletitle{Routing of {AGVs} on automated container
  terminals}. In \bibinfo{booktitle}{\emph{2015 {IEEE} 19th {International}
  {Conference} on {Computer} {Supported} {Cooperative} {Work} in {Design}
  ({CSCWD})}}. \bibinfo{publisher}{IEEE}, \bibinfo{address}{Calabria, Italy},
  \bibinfo{pages}{401--406}.
\newblock
\href{https://doi.org/10.1109/CSCWD.2015.7230993}{doi:\nolinkurl{10.1109/CSCWD.2015.7230993}}


\bibitem[Gong et~al\mbox{.}(2023)]%
        {gongSafeEfficientHuman2023}
\bibfield{author}{\bibinfo{person}{Xin Gong}, \bibinfo{person}{Tieniu Wang},
  \bibinfo{person}{Tingwen Huang}, {and} \bibinfo{person}{Yukang Cui}.}
  \bibinfo{year}{2023}\natexlab{}.
\newblock \showarticletitle{Toward {Safe} and {Efficient} {Human}–{Swarm}
  {Collaboration}: {A} {Hierarchical} {Multi}-{Agent} {Pickup} and {Delivery}
  {Framework}}.
\newblock \bibinfo{journal}{\emph{IEEE Transactions on Intelligent Vehicles}}
  \bibinfo{volume}{8}, \bibinfo{number}{2} (\bibinfo{date}{Feb.}
  \bibinfo{year}{2023}), \bibinfo{pages}{1664--1675}.
\newblock
\showISSN{2379-8904}
\href{https://doi.org/10.1109/TIV.2022.3172342}{doi:\nolinkurl{10.1109/TIV.2022.3172342}}


\bibitem[Hart et~al\mbox{.}(1968)]%
        {Hart1968Astar}
\bibfield{author}{\bibinfo{person}{Peter~E. Hart}, \bibinfo{person}{Nils~J.
  Nilsson}, {and} \bibinfo{person}{Bertram Raphael}.}
  \bibinfo{year}{1968}\natexlab{}.
\newblock \showarticletitle{A Formal Basis for the Heuristic Determination of
  Minimum Cost Paths}.
\newblock \bibinfo{journal}{\emph{IEEE Transactions on Systems Science and
  Cybernetics}} \bibinfo{volume}{4}, \bibinfo{number}{2}
  (\bibinfo{year}{1968}), \bibinfo{pages}{100--107}.
\newblock
\href{https://doi.org/10.1109/TSSC.1968.300136}{doi:\nolinkurl{10.1109/TSSC.1968.300136}}


\bibitem[Hönig et~al\mbox{.}(2018)]%
        {honigConflictbasedSearchOptimal2018}
\bibfield{author}{\bibinfo{person}{Wolfgang Hönig}, \bibinfo{person}{Scott
  Kiesel}, \bibinfo{person}{Andrew Tinka}, \bibinfo{person}{Joseph~W. Durham},
  {and} \bibinfo{person}{Nora Ayanian}.} \bibinfo{year}{2018}\natexlab{}.
\newblock \showarticletitle{Conflict-based search with optimal task
  assignment}. In \bibinfo{booktitle}{\emph{Proceedings of the 17th
  international conference on autonomous agents and {MultiAgent} systems}}
  \emph{(\bibinfo{series}{Aamas '18})}. \bibinfo{publisher}{International
  Foundation for Autonomous Agents and Multiagent Systems},
  \bibinfo{address}{Richland, SC}, \bibinfo{pages}{757--765}.
\newblock


\bibitem[Kou et~al\mbox{.}(2020)]%
        {kouIdleTimeOptimization2020}
\bibfield{author}{\bibinfo{person}{Ngai~Meng Kou}, \bibinfo{person}{Cheng
  Peng}, \bibinfo{person}{Hang Ma}, \bibinfo{person}{T.~K.~Satish Kumar}, {and}
  \bibinfo{person}{Sven Koenig}.} \bibinfo{year}{2020}\natexlab{}.
\newblock \showarticletitle{Idle {Time} {Optimization} for {Target}
  {Assignment} and {Path} {Finding} in {Sortation} {Centers}}.
\newblock \bibinfo{journal}{\emph{Proceedings of the AAAI Conference on
  Artificial Intelligence}} \bibinfo{volume}{34}, \bibinfo{number}{06}
  (\bibinfo{date}{April} \bibinfo{year}{2020}), \bibinfo{pages}{9925--9932}.
\newblock
\showISSN{2374-3468, 2159-5399}
\href{https://doi.org/10.1609/aaai.v34i06.6547}{doi:\nolinkurl{10.1609/aaai.v34i06.6547}}


\bibitem[Krnjaic et~al\mbox{.}(2022)]%
        {krnjaicScalableMultiAgentReinforcement2022}
\bibfield{author}{\bibinfo{person}{Aleksandar Krnjaic},
  \bibinfo{person}{Jonathan~D. Thomas}, \bibinfo{person}{Georgios Papoudakis},
  \bibinfo{person}{Lukas Schäfer}, \bibinfo{person}{Peter Börsting}, {and}
  \bibinfo{person}{Stefano~V. Albrecht}.} \bibinfo{year}{2022}\natexlab{}.
\newblock \bibinfo{title}{Scalable {Multi}-{Agent} {Reinforcement} {Learning}
  for {Warehouse} {Logistics} with {Robotic} and {Human} {Co}-{Workers}}.
\newblock
\href{https://doi.org/10.48550/arXiv.2212.11498}{doi:\nolinkurl{10.48550/arXiv.2212.11498}}


\bibitem[Lau and Sengupta(2022)]%
        {lauMultiAgentPickupDelivery2022}
\bibfield{author}{\bibinfo{person}{Tim Tsz-Kit Lau} {and}
  \bibinfo{person}{Biswa Sengupta}.} \bibinfo{year}{2022}\natexlab{}.
\newblock \bibinfo{title}{The {Multi}-{Agent} {Pickup} and {Delivery}
  {Problem}: {MAPF}, {MARL} and {Its} {Warehouse} {Applications}}.
\newblock
\href{https://doi.org/10.48550/arXiv.2203.07092}{doi:\nolinkurl{10.48550/arXiv.2203.07092}}


\bibitem[Li et~al\mbox{.}(2021a)]%
        {liAnytimeMultiAgentPath2021}
\bibfield{author}{\bibinfo{person}{Jiaoyang Li}, \bibinfo{person}{Zhe Chen},
  \bibinfo{person}{Daniel Harabor}, \bibinfo{person}{Peter~J. Stuckey}, {and}
  \bibinfo{person}{Sven Koenig}.} \bibinfo{year}{2021}\natexlab{a}.
\newblock \showarticletitle{Anytime {Multi}-{Agent} {Path} {Finding} via
  {Large} {Neighborhood} {Search}}. In \bibinfo{booktitle}{\emph{Proceedings of
  the Thirtieth International Joint Conference on Artificial Intelligence,
  {IJCAI-21}}}, Vol.~\bibinfo{volume}{4}. \bibinfo{publisher}{International
  Joint Conferences on Artificial Intelligence Organization},
  \bibinfo{address}{Montreal, Canada}, \bibinfo{pages}{4127--4135}.
\newblock
\href{https://doi.org/10.24963/ijcai.2021/568}{doi:\nolinkurl{10.24963/ijcai.2021/568}}


\bibitem[Li et~al\mbox{.}(2022)]%
        {liMAPFLNS2FastRepairing2022}
\bibfield{author}{\bibinfo{person}{Jiaoyang Li}, \bibinfo{person}{Zhe Chen},
  \bibinfo{person}{Daniel Harabor}, \bibinfo{person}{Peter~J. Stuckey}, {and}
  \bibinfo{person}{Sven Koenig}.} \bibinfo{year}{2022}\natexlab{}.
\newblock \showarticletitle{{MAPF}-{LNS2}: {Fast} {Repairing} for
  {Multi}-{Agent} {Path} {Finding} via {Large} {Neighborhood} {Search}}.
\newblock \bibinfo{journal}{\emph{Proceedings of the AAAI Conference on
  Artificial Intelligence}} \bibinfo{volume}{36}, \bibinfo{number}{9}
  (\bibinfo{date}{June} \bibinfo{year}{2022}), \bibinfo{pages}{10256--10265}.
\newblock
\showISSN{2374-3468}
\href{https://doi.org/10.1609/aaai.v36i9.21266}{doi:\nolinkurl{10.1609/aaai.v36i9.21266}}


\bibitem[Li et~al\mbox{.}(2021b)]%
        {liPairwiseSymmetryReasoning2021}
\bibfield{author}{\bibinfo{person}{Jiaoyang Li}, \bibinfo{person}{Daniel
  Harabor}, \bibinfo{person}{Peter~J. Stuckey}, \bibinfo{person}{Hang Ma},
  \bibinfo{person}{Graeme Gange}, {and} \bibinfo{person}{Sven Koenig}.}
  \bibinfo{year}{2021}\natexlab{b}.
\newblock \showarticletitle{Pairwise symmetry reasoning for multi-agent path
  finding search}.
\newblock \bibinfo{journal}{\emph{Artificial Intelligence}}
  \bibinfo{volume}{301} (\bibinfo{date}{Dec.} \bibinfo{year}{2021}),
  \bibinfo{pages}{103574}.
\newblock
\showISSN{0004-3702}
\href{https://doi.org/10.1016/j.artint.2021.103574}{doi:\nolinkurl{10.1016/j.artint.2021.103574}}


\bibitem[Li et~al\mbox{.}(2021c)]%
        {liEECBSBoundedSuboptimalSearch2021}
\bibfield{author}{\bibinfo{person}{Jiaoyang Li}, \bibinfo{person}{Wheeler
  Ruml}, {and} \bibinfo{person}{Sven Koenig}.}
  \bibinfo{year}{2021}\natexlab{c}.
\newblock \bibinfo{title}{{EECBS}: {A} {Bounded}-{Suboptimal} {Search} for
  {Multi}-{Agent} {Path} {Finding}}.
\newblock
\href{https://doi.org/10.48550/arXiv.2010.01367}{doi:\nolinkurl{10.48550/arXiv.2010.01367}}


\bibitem[Li et~al\mbox{.}(2021d)]%
        {liLifelongMultiAgentPath2021}
\bibfield{author}{\bibinfo{person}{Jiaoyang Li}, \bibinfo{person}{Andrew
  Tinka}, \bibinfo{person}{Scott Kiesel}, \bibinfo{person}{Joseph~W. Durham},
  \bibinfo{person}{T.~K.~Satish Kumar}, {and} \bibinfo{person}{Sven Koenig}.}
  \bibinfo{year}{2021}\natexlab{d}.
\newblock \bibinfo{title}{Lifelong {Multi}-{Agent} {Path} {Finding} in
  {Large}-{Scale} {Warehouses}}.
\newblock
\href{https://doi.org/10.48550/arXiv.2005.07371}{doi:\nolinkurl{10.48550/arXiv.2005.07371}}


\bibitem[Liu et~al\mbox{.}(2019)]%
        {liuTaskPathPlanning2019}
\bibfield{author}{\bibinfo{person}{Minghua Liu}, \bibinfo{person}{Hang Ma},
  \bibinfo{person}{Jiaoyang Li}, {and} \bibinfo{person}{Sven Koenig}.}
  \bibinfo{year}{2019}\natexlab{}.
\newblock \showarticletitle{Task and path planning for multi-agent pickup and
  delivery}. In \bibinfo{booktitle}{\emph{Proceedings of the 18th
  {International} {Conference} on {Autonomous} {Agents} and {MultiAgent}
  {Systems}}} \emph{(\bibinfo{series}{{AAMAS} '19})}.
  \bibinfo{publisher}{International Foundation for Autonomous Agents and
  Multiagent Systems}, \bibinfo{address}{Richland, SC},
  \bibinfo{pages}{1152--1160}.
\newblock
\showISBNx{978-1-4503-6309-9}


\bibitem[Lodigiani et~al\mbox{.}(2023)]%
        {lodigianiRobustMultiagentPickup2023}
\bibfield{author}{\bibinfo{person}{Giacomo Lodigiani}, \bibinfo{person}{Nicola
  Basilico}, {and} \bibinfo{person}{Francesco Amigoni}.}
  \bibinfo{year}{2023}\natexlab{}.
\newblock \showarticletitle{Robust Multi-Agent Pickup and Delivery with
  Delays}. In \bibinfo{booktitle}{\emph{2023 {{European Conference}} on
  {{Mobile Robots}} ({{ECMR}})}} (2023-09). \bibinfo{publisher}{IEEE},
  \bibinfo{address}{Coimbra, Portugal}, \bibinfo{pages}{1--8}.
\newblock
\showISSN{2767-8733}
\href{https://doi.org/10.1109/ECMR59166.2023.10256381}{doi:\nolinkurl{10.1109/ECMR59166.2023.10256381}}


\bibitem[Loganathan and Ahmad(2023)]%
        {loganathanSystematicReviewRecent2023}
\bibfield{author}{\bibinfo{person}{Anbalagan Loganathan} {and}
  \bibinfo{person}{Nur~Syazreen Ahmad}.} \bibinfo{year}{2023}\natexlab{}.
\newblock \showarticletitle{A systematic review on recent advances in
  autonomous mobile robot navigation}.
\newblock \bibinfo{journal}{\emph{Engineering Science and Technology, an
  International Journal}}  \bibinfo{volume}{40} (\bibinfo{date}{April}
  \bibinfo{year}{2023}), \bibinfo{pages}{101343}.
\newblock
\showISSN{2215-0986}
\href{https://doi.org/10.1016/j.jestch.2023.101343}{doi:\nolinkurl{10.1016/j.jestch.2023.101343}}


\bibitem[Ma et~al\mbox{.}(2017)]%
        {maLifelongMultiagentPath2017}
\bibfield{author}{\bibinfo{person}{Hang Ma}, \bibinfo{person}{Jiaoyang Li},
  \bibinfo{person}{T.K.~Satish Kumar}, {and} \bibinfo{person}{Sven Koenig}.}
  \bibinfo{year}{2017}\natexlab{}.
\newblock \showarticletitle{Lifelong multi-agent path finding for online pickup
  and delivery tasks}. In \bibinfo{booktitle}{\emph{Proceedings of the 16th
  conference on autonomous agents and {MultiAgent} systems}}
  \emph{(\bibinfo{series}{Aamas '17})}. \bibinfo{publisher}{International
  Foundation for Autonomous Agents and Multiagent Systems},
  \bibinfo{address}{Richland, SC}, \bibinfo{pages}{837--845}.
\newblock


\bibitem[Ma et~al\mbox{.}(2016)]%
        {maMultiAgentPathFinding2016}
\bibfield{author}{\bibinfo{person}{Hang Ma}, \bibinfo{person}{Craig Tovey},
  \bibinfo{person}{Guni Sharon}, \bibinfo{person}{T.~K. Kumar}, {and}
  \bibinfo{person}{Sven Koenig}.} \bibinfo{year}{2016}\natexlab{}.
\newblock \showarticletitle{Multi-{Agent} {Path} {Finding} with {Payload}
  {Transfers} and the {Package}-{Exchange} {Robot}-{Routing} {Problem}}.
\newblock \bibinfo{journal}{\emph{Proceedings of the AAAI Conference on
  Artificial Intelligence}} \bibinfo{volume}{30}, \bibinfo{number}{1}
  (\bibinfo{date}{March} \bibinfo{year}{2016}), \bibinfo{pages}{3166--3173}.
\newblock
\href{https://doi.org/10.1609/aaai.v30i1.10409}{doi:\nolinkurl{10.1609/aaai.v30i1.10409}}


\bibitem[Nguyen et~al\mbox{.}(2017)]%
        {nguyenGeneralizedTargetAssignment2017}
\bibfield{author}{\bibinfo{person}{Van Nguyen}, \bibinfo{person}{Philipp
  Obermeier}, \bibinfo{person}{Tran~Cao Son}, \bibinfo{person}{Torsten Schaub},
  {and} \bibinfo{person}{William Yeoh}.} \bibinfo{year}{2017}\natexlab{}.
\newblock \showarticletitle{Generalized target assignment and path finding
  using answer set programming}. In \bibinfo{booktitle}{\emph{Proceedings of
  the twenty-sixth international joint conference on artificial intelligence,
  {IJCAI}-17}}. \bibinfo{publisher}{International Joint Conferences on
  Artificial Intelligence Organization}, \bibinfo{address}{Melbourne,
  Australia}, \bibinfo{pages}{1216--1223}.
\newblock
\href{https://doi.org/10.24963/ijcai.2017/169}{doi:\nolinkurl{10.24963/ijcai.2017/169}}


\bibitem[Okumura(2022)]%
        {okumuraLaCAMSearchBasedAlgorithm2022}
\bibfield{author}{\bibinfo{person}{Keisuke Okumura}.}
  \bibinfo{year}{2022}\natexlab{}.
\newblock \bibinfo{title}{{LaCAM}: {Search}-{Based} {Algorithm} for {Quick}
  {Multi}-{Agent} {Pathfinding}}.
\newblock
\href{https://doi.org/10.48550/arXiv.2211.13432}{doi:\nolinkurl{10.48550/arXiv.2211.13432}}


\bibitem[Okumura(2023)]%
        {okumuraImprovingLaCAMScalable2023}
\bibfield{author}{\bibinfo{person}{Keisuke Okumura}.}
  \bibinfo{year}{2023}\natexlab{}.
\newblock \bibinfo{title}{Improving {LaCAM} for {Scalable} {Eventually}
  {Optimal} {Multi}-{Agent} {Pathfinding}}.
\newblock
\href{https://doi.org/10.48550/arXiv.2305.03632}{doi:\nolinkurl{10.48550/arXiv.2305.03632}}


\bibitem[Okumura and D{\'{e}}fago(2022)]%
        {okumuraSolvingSimultaneousTarget2022}
\bibfield{author}{\bibinfo{person}{Keisuke Okumura} {and}
  \bibinfo{person}{Xavier D{\'{e}}fago}.} \bibinfo{year}{2022}\natexlab{}.
\newblock \showarticletitle{Solving Simultaneous Target Assignment and Path
  Planning Efficiently with Time-Independent Execution}. In
  \bibinfo{booktitle}{\emph{Proceedings of the Thirty-Second International
  Conference on Automated Planning and Scheduling, {ICAPS} 2022, Singapore
  (virtual), June 13-24, 2022}}, \bibfield{editor}{\bibinfo{person}{Akshat
  Kumar}, \bibinfo{person}{Sylvie Thi{\'{e}}baux}, \bibinfo{person}{Pradeep
  Varakantham}, {and} \bibinfo{person}{William Yeoh}} (Eds.).
  \bibinfo{publisher}{{AAAI} Press}, \bibinfo{pages}{270--278}.
\newblock
\href{https://doi.org/10.1609/icaps.v32i1.19810}{doi:\nolinkurl{10.1609/icaps.v32i1.19810}}


\bibitem[Sharon et~al\mbox{.}(2012a)]%
        {sharonConflictBasedSearchOptimal2012}
\bibfield{author}{\bibinfo{person}{Guni Sharon}, \bibinfo{person}{Roni Stern},
  \bibinfo{person}{Ariel Felner}, {and} \bibinfo{person}{Nathan Sturtevant}.}
  \bibinfo{year}{2012}\natexlab{a}.
\newblock \showarticletitle{Conflict-{Based} {Search} {For} {Optimal}
  {Multi}-{Agent} {Path} {Finding}}.
\newblock \bibinfo{journal}{\emph{Proceedings of the AAAI Conference on
  Artificial Intelligence}} \bibinfo{volume}{26}, \bibinfo{number}{1}
  (\bibinfo{year}{2012}), \bibinfo{pages}{563--569}.
\newblock
\showISSN{2374-3468}
\href{https://doi.org/10.1609/aaai.v26i1.8140}{doi:\nolinkurl{10.1609/aaai.v26i1.8140}}


\bibitem[Sharon et~al\mbox{.}(2012b)]%
        {sharonMetaAgentConflictBasedSearch2012}
\bibfield{author}{\bibinfo{person}{Guni Sharon}, \bibinfo{person}{Roni Stern},
  \bibinfo{person}{Ariel Felner}, {and} \bibinfo{person}{Nathan Sturtevant}.}
  \bibinfo{year}{2012}\natexlab{b}.
\newblock \showarticletitle{Meta-{Agent} {Conflict}-{Based} {Search} {For}
  {Optimal} {Multi}-{Agent} {Path} {Finding}}.
\newblock \bibinfo{journal}{\emph{Proceedings of the International Symposium on
  Combinatorial Search}} \bibinfo{volume}{3}, \bibinfo{number}{1}
  (\bibinfo{year}{2012}), \bibinfo{pages}{97--104}.
\newblock
\showISSN{2832-9163}
\href{https://doi.org/10.1609/socs.v3i1.18244}{doi:\nolinkurl{10.1609/socs.v3i1.18244}}


\bibitem[Shen et~al\mbox{.}(2023)]%
        {shenTrackingProgressMultiAgent2023}
\bibfield{author}{\bibinfo{person}{Bojie Shen}, \bibinfo{person}{Zhe Chen},
  \bibinfo{person}{Muhammad~Aamir Cheema}, \bibinfo{person}{Daniel~D. Harabor},
  {and} \bibinfo{person}{Peter~J. Stuckey}.} \bibinfo{year}{2023}\natexlab{}.
\newblock \bibinfo{title}{Tracking {Progress} in {Multi}-{Agent} {Path}
  {Finding}}.
\newblock
\href{https://doi.org/10.48550/arXiv.2305.08446}{doi:\nolinkurl{10.48550/arXiv.2305.08446}}


\bibitem[Stern et~al\mbox{.}(2019)]%
        {sternMultiAgentPathfindingDefinitions2019}
\bibfield{author}{\bibinfo{person}{Roni Stern}, \bibinfo{person}{Nathan
  Sturtevant}, \bibinfo{person}{Ariel Felner}, \bibinfo{person}{Sven Koenig},
  \bibinfo{person}{Hang Ma}, \bibinfo{person}{Thayne Walker},
  \bibinfo{person}{Jiaoyang Li}, \bibinfo{person}{Dor Atzmon},
  \bibinfo{person}{Liron Cohen}, \bibinfo{person}{T.~K.~Satish Kumar},
  \bibinfo{person}{Eli Boyarski}, {and} \bibinfo{person}{Roman Bartak}.}
  \bibinfo{year}{2019}\natexlab{}.
\newblock \showarticletitle{Multi-{Agent} {Pathfinding}: {Definitions},
  {Variants}, and {Benchmarks}}.
\newblock \bibinfo{journal}{\emph{arXiv:1906.08291 [cs]}} (\bibinfo{date}{June}
  \bibinfo{year}{2019}).
\newblock
\urldef\tempurl%
\url{http://arxiv.org/abs/1906.08291}
\showURL{%
\tempurl}


\bibitem[Wang et~al\mbox{.}(2023)]%
        {wangOrderMattersAgentbyagent2023}
\bibfield{author}{\bibinfo{person}{Xihuai Wang}, \bibinfo{person}{Zheng Tian},
  \bibinfo{person}{Ziyu Wan}, \bibinfo{person}{Ying Wen}, \bibinfo{person}{Jun
  Wang}, {and} \bibinfo{person}{Weinan Zhang}.}
  \bibinfo{year}{2023}\natexlab{}.
\newblock \showarticletitle{Order {Matters}: {Agent}-by-agent {Policy}
  {Optimization}}. In \bibinfo{booktitle}{\emph{The Eleventh International
  Conference on Learning Representations, {ICLR} 2023, Kigali, Rwanda, May 1-5,
  2023}}. \bibinfo{publisher}{OpenReview.net}.
\newblock
\urldef\tempurl%
\url{https://openreview.net/forum?id=Q-neeWNVv1}
\showURL{%
\tempurl}


\bibitem[Xu et~al\mbox{.}(2022)]%
        {xuMultigoalMultiagentPickup2022}
\bibfield{author}{\bibinfo{person}{Qinghong Xu}, \bibinfo{person}{Jiaoyang Li},
  \bibinfo{person}{Sven Koenig}, {and} \bibinfo{person}{Hang Ma}.}
  \bibinfo{year}{2022}\natexlab{}.
\newblock \showarticletitle{Multi-goal multi-agent pickup and delivery}. In
  \bibinfo{booktitle}{\emph{2022 {IEEE}/{RSJ} international conference on
  intelligent robots and systems ({IROS})}}. \bibinfo{publisher}{IEEE},
  \bibinfo{address}{Kyoto, Japan}, \bibinfo{pages}{9964--9971}.
\newblock
\href{https://doi.org/10.1109/IROS47612.2022.9981785}{doi:\nolinkurl{10.1109/IROS47612.2022.9981785}}


\bibitem[Zhang et~al\mbox{.}(2018)]%
        {zhangCollisionfreeRoutePlanning2018}
\bibfield{author}{\bibinfo{person}{Zheng Zhang}, \bibinfo{person}{Qing Guo},
  \bibinfo{person}{Juan Chen}, {and} \bibinfo{person}{Peijiang Yuan}.}
  \bibinfo{year}{2018}\natexlab{}.
\newblock \showarticletitle{Collision-free route planning for multiple {AGVs}
  in an automated warehouse based on collision classification}.
\newblock \bibinfo{journal}{\emph{IEEE access : practical innovations, open
  solutions}}  \bibinfo{volume}{6} (\bibinfo{year}{2018}),
  \bibinfo{pages}{26022--26035}.
\newblock
\href{https://doi.org/10.1109/ACCESS.2018.2819199}{doi:\nolinkurl{10.1109/ACCESS.2018.2819199}}


\end{thebibliography}

\clearpage
\appendix

\section{Real-world Demonstration}
\label{app:photographs}

Our MAPDwHC formulation is inspired by modern e-commerce fulfillment centers that stock thousands of stock-keeping units (SKUs).
Because the shelving is non-standardized and items vary widely in size, shape, and packaging, current robotic grippers cannot manipulate every product with high reliability.
As a result, human workers perform the \textit{pick} and \textit{pack} operations, while fleets of mobile robots ferry plastic totes between storage aisles and packing stations.
Each robot is equipped with on-board LiDAR, localizes itself using SLAM, and continuously transmits its position to a central planner.
The robots reliably detect and avoid nearby humans, and workers can easily step through the gap between a robot and the shelving; hence collisions between humans and robots are negligible.
Accordingly, our MAPDwHC model assumes that human-robot collisions do not occur.

\begin{figure}[!t]
  \centering
  \includegraphics[width=\columnwidth]{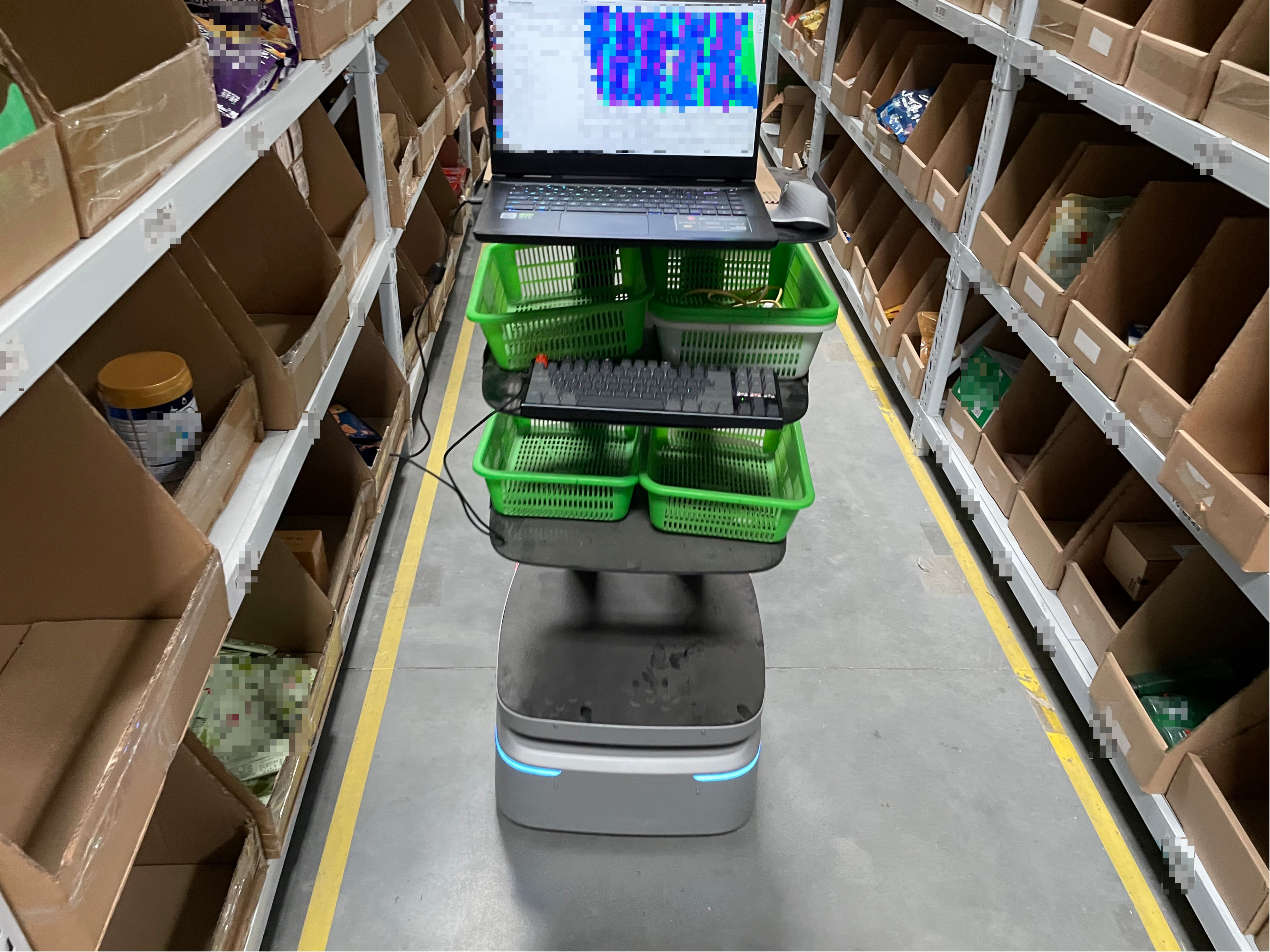}
  \caption{The autonomous mobile robots operating in the warehouse transport plastic totes that serve as carriers for customer orders. Each robot can carry multiple SKUs at once. Because the robots are not equipped with grippers, human workers manually place items into the totes. The merchandise is stored on shelving units and exhibits considerable variation in size, shape, and packaging.}
  \label{fig:real_robot}
  \Description{The autonomous mobile robots operating in the warehouse transport plastic totes that serve as carriers for customer orders. Each robot can carry multiple SKUs at once. Because the robots are not equipped with grippers, human workers manually place items into the totes. The merchandise is stored on shelving units and exhibits considerable variation in size, shape, and packaging.}
\end{figure}

\begin{figure}[!t]
  \centering
  \includegraphics[width=\columnwidth]{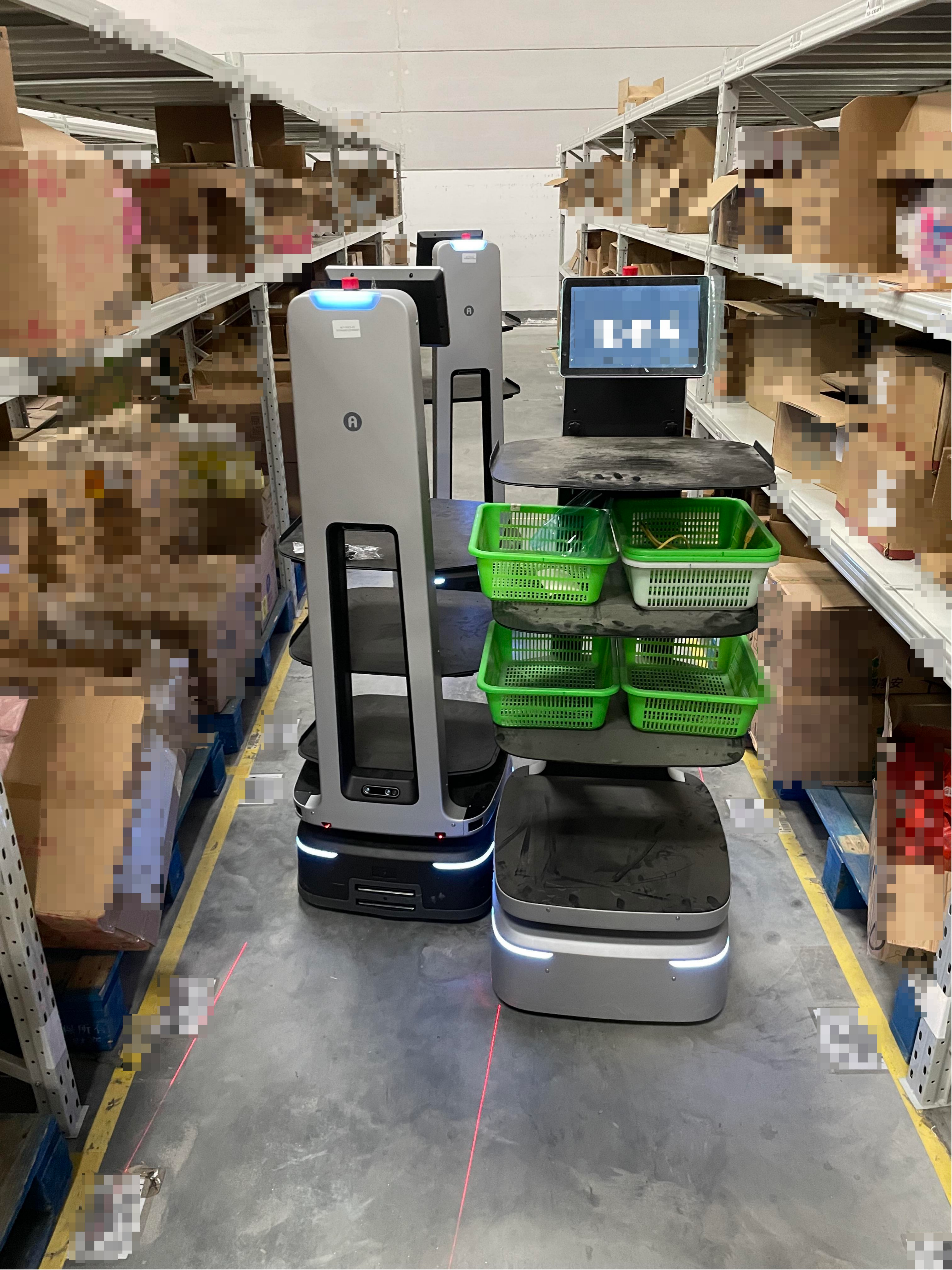}
  \caption{Coordination in narrow aisles. During the picking phase, the tote-carrying robot in the foreground remains stationary, allowing the human picker to quickly place items into its bins. Two autonomous robots in the background detect this stationary robot and navigate around it while maintaining a safe clearance. Accordingly, the tote-carrying robot is treated as a static obstacle in our MAPDwHC formulation during human-robot interaction.}
  \label{fig:robot_yielding}
  \Description{Coordination in narrow aisles. During the picking phase, the tote-carrying robot in the foreground remains stationary, allowing the human picker to quickly place items into its bins. Two autonomous robots in the background detect this stationary robot and navigate around it while maintaining a safe clearance. Accordingly, the tote-carrying robot is treated as a static obstacle in our MAPDwHC formulation during human-robot interaction.}
\end{figure}

The photographs in Figures~\ref{fig:real_robot} and \ref{fig:robot_yielding} highlight the practical constraints shaping our design: narrow aisles, close human-robot interaction, and limited payload capacity.
Collectively, these challenges motivate the fast, reactive algorithms detailed in the main paper.

\section{Extended Proof of Theorem~\ref{th:1}}
\label{app:extended_proof}
\begin{proof}
    At any time step $t$, we select a subset of at most $|A|$ pending tasks and assign each one to a distinct idle agent ($c_i < \mathcal{K}$). When the number of pending tasks exceeds $|A|$, the remaining tasks are deferred to subsequent time steps.
    Treating every agent's current vertex as its start and the corresponding pickup or delivery vertex as its goal yields a valid MAPF instance.

    When no human-robot interaction is occurring, all robots are free to move.
    Property~\ref{property:map_connectivity} guarantees that every goal is reachable from any vertex, and Property~\ref{property:map_layout} lets us swap agents via the buffer vertices in $V_F$ whenever a path is blocked.
    Together, these properties ensure that a complete MAPF solver can always find a collision-free plan.

    If a human-robot interaction is underway, we remove the interacting robot's occupied vertex from $G$ to obtain the reduced graph $G'=(V \setminus \{v_t^{a_i}\}, E')$.
    This deletion may temporarily violate Property~\ref{property:map_connectivity} or Property~\ref{property:map_layout}, rendering the MAPF instance unsolvable.
    However, Property~\ref{property:finite_interaction} bounds the duration of every interaction, so the blocked vertex is released after finite time steps and then the connectivity is restored.
    All agents can simply wait until the interaction ends.

    Consequently, every task in $\mathcal{T}$ is eventually assigned to an agent, and the set of tasks selected at each iteration constitutes a solvable MAPF sub-problem. Every task is completed within finite time.
    Thus, any well-formed MAPDwHC instance is solvable.
\end{proof}

\section{Comparison of MAPF Algorithms}
\label{app:mapf}

Table~\ref{tab:algo_performance} compares the performance of different MAPF algorithms on the \texttt{warehouse\_large} map with varying agent counts, using the CARVS over 10,000 time steps.
CBSH2RTC fails with 60 and 80 agents.
The memory usage of CBSH2RTC will increase tremendously (e.g. 200GB for 60 agents), eventually terminated by the operating system.
Both CBSH2RTC, EECBS and LaCAM struggle as agent numbers increase.
LaCAM2 outperforms the others in completed tasks, making it our preferred low-level MAPF algorithm.
Although MAPFLNS2 and RHCR can successfully solve the problem, their performance is slightly below that of LaCAM2 as the number of agents grows.

All MAPF algorithms were run with their default parameters unless stated otherwise.
Specifically, EECBS uses a suboptimality factor of 1.2, and RHCR is configured with a simulation time of 5, a simulation window of 1, and a planning window of 10.

\begin{table}[!t]
  \centering
  \caption{MAPF Algorithms Performance by Number of Agents}
  \begin{tabular}{lcc}
  \hline
  Algorithm & Agents & Tasks Completed \\
  \hline
  \multirow{3}{*}{CBSH2RTC} & 40 & 25±70 \\
                           & 60 & 0 \\
                           & 80 & 0 \\
  \hline
  \multirow{3}{*}{EECBS}   & 40 & 4,938±4,310 \\
                           & 60 & 775±820 \\
                           & 80 & 101±215 \\
  \hline
  \multirow{3}{*}{LaCAM}   & 40 & 3,166±630 \\
                           & 60 & 1,355±865 \\
                           & 80 & 360±300 \\
  \hline
  \multirow{3}{*}{LaCAM2}  & 40 & 8,891±260 \\
                           & 60 & 11,731±570 \\
                           & 80 & \textbf{13,966±455} \\
  \hline
  \multirow{3}{*}{MAPFLNS2}& 40 & 9,060±260 \\
                           & 60 & 10,927±430 \\
                           & 80 & 12,150±280 \\
  \hline
  \multirow{3}{*}{RHCR}    & 40 & 10,033±230 \\
                           & 60 & 11,827±415 \\
                           & 80 & 12,020±760 \\
  \hline
  \end{tabular}
  \label{tab:algo_performance}
\end{table}

\section{Details of Compared Algorithms}
\label{appendix:tpts_adaptation}

We modified TP/TPTS and RMCA/MCA so that they operate under our MAPDwHC assumptions while preserving their original decision logic.
First, TP/TPTS now runs in an online manner: at every time step it returns a partial plan that ends when the first task of the current task sets finishes, rather than waiting until all tasks are completed.
The simulator verifies this plan, manages the pick-up and drop-off states, and injects the drop-off vertices. Each drop-off task carries an agent-specific ID to prevent reassignment.

RMCA/MCA constructs a heap of all tasks and agents at the start of the simulation before producing the paths.
We trigger a global replan as soon as any agent reaches its goal. Drop-off locations are also pre-assigned to the corresponding agents and treated as exclusive goals.

Whenever an agent engages in human interaction its current vertex is marked temporarily as an obstacle.
Our simulation removes that agent from the active set and alters the graph supplied to the solver.
Any change in task status or human-agent interaction immediately triggers a new round of task allocation and path planning, ensuring that the planners face the same dynamic conditions as CARVS and the other algorithms.
These adaptations allow TP, TPTS, RMCA, and MCA to be compared fairly with the proposed methods in realistic warehouse scenarios that involve human-robot collaboration.

TSWAP was originally designed for simultaneous assignment to unlabeled targets.
In our MAPD setting, any robot can be assigned to a pickup target, but a robot carrying an item must reach its assigned drop-off target.
To apply TSWAP in this setting, the wrapper supplies current pickup targets to picking robots and assigned drop-off targets to dropping robots when constructing each planning instance.
Each drop-off target carries the corresponding agent index, so the initial target allocation preserves the assignment.
However, TSWAP handles conflicts by swapping targets and breaks deadlocks by rotating them.
These operations change existing task assignments, so TSWAP cannot maintain a fixed assignment between a robot and its target during execution.
A dropping robot may therefore be redirected toward another drop-off target or a pickup target.
The wrapper can only reapply the original assignment in a later planning attempt.

Diagnostic results on \texttt{warehouse\_large} with 100 robots show how these target exchanges affect task completion.
Only 0.73\% of successful paths for dropping robots ended at the assigned drop-off target.
Instead, 44.94\% ended at another robot's drop-off target and 54.33\% at a pickup target.
Neither pathfinding failures nor solver execution errors were observed.
Robots spent an average of 45.92\% of their time in the drop-off phase.
The median duration of a completed drop-off averaged 728 steps across runs.
When the simulation ended, an average of 48.6 robots had not completed their drop-offs.
These results show that fast planning and a low skip count do not guarantee efficient task completion.

MAPD-d exposes a robustness parameter $k$, which stands for the number of delay steps an agent must absorb without replanning.
However, our setting provides no reliable value for it.
Although the interaction time is set to 10 time steps in the simulations, the picker still needs time to walk to the picking point, and that walking time varies with the distance.
So the delay an agent experiences is unknown in advance.
We therefore report $k=0$ in the main results.
We also ran $k=10$ for comparison.
On \texttt{warehouse\_small} it completes 989 tasks at 20 agents against 3703 for $k=0$, and it completes no tasks at all beyond 40 agents.
The reason is that $k$-TP reserves a temporal window of $\pm k$ around every cell on an already planned path, so each agent contributes $2k+1$ space-time reservations instead of one.
A larger $k$ therefore inflates the set of moving obstacles the low-level search must avoid, which raises the planning time and makes the real-time budget harder to meet.

\section{Human Efficiency}
\label{app:human_eff}

Figure~\ref{fig:es_not_found} shows how human-efficiency settings affect the number of skips on \texttt{warehouse\_small} as the number of robots increases from 10 to 100.
With the reaction time fixed at $\eta = 1\,\text{s}$, reducing the interaction duration $\delta$ and increasing human walking speed both lower the skip count and thus boost overall performance.
For example, configuration \texttt{D5S2} imposes a 5-step interaction delay while letting humans move twice as fast as robots, whereas \texttt{D0S100} approximates the original MAPD problem setting by introducing virtually no delay and letting humans move far faster than the robots.
As $\delta$ decreases the skip count falls sharply.
In contrast, changes in human walking speed have little effect, with the skip count remaining nearly unchanged between \texttt{D0S2} and \texttt{D0S100}.
Comparing \texttt{D10S2}, the configuration used in our experiments, with \texttt{D0S2} shows that MAPDwHC records about $1.96\times$ more skips than MAPD, highlighting the extra difficulty introduced by human-robot interaction.

\begin{figure}
  \centering
  \includegraphics[width=0.95\linewidth]{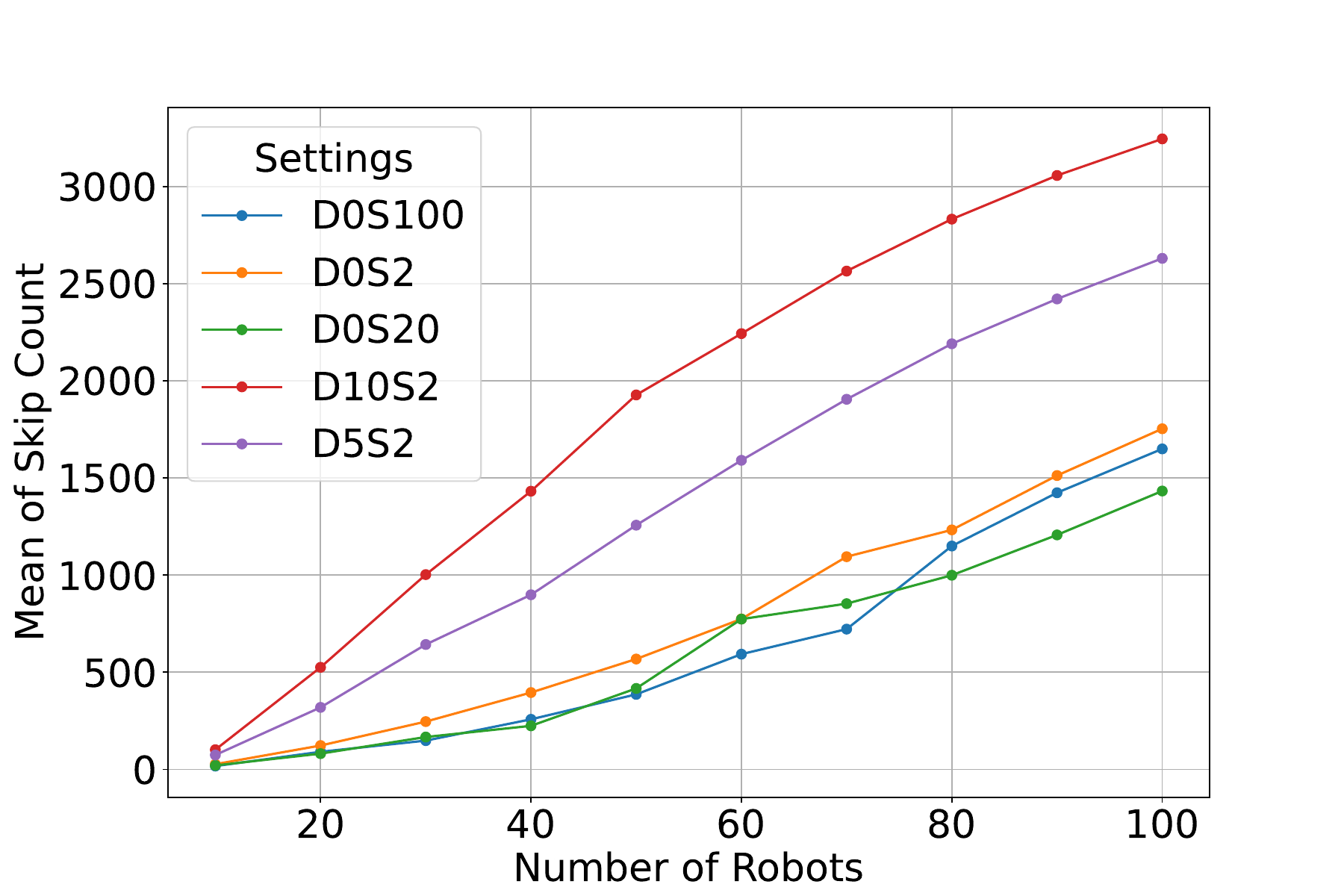}
  \caption{Impact of human-efficiency parameters on the performance of CARVS with LaCAM2 on \texttt{warehouse\_small} for 10 to 100 robots. The reaction time $\eta$ is fixed at 1 second. Configuration \texttt{D5S2} sets the interaction duration to $\delta = 5$ time steps and assigns humans a speed twice that of the robots. Reducing the interaction duration substantially decreases the number of skips.}
  \Description{Impact of human-efficiency parameters on the performance of CARVS with LaCAM2 on \texttt{warehouse\_small} for 10 to 100 robots. The reaction time $\eta$ is fixed at 1 second. Configuration \texttt{D5S2} sets the interaction duration to $\delta = 5$ time steps and assigns humans a speed twice that of the robots. Reducing the interaction duration substantially decreases the number of skips.}
  \label{fig:es_not_found}
\end{figure}

Figure~\ref{fig:human_num} shows how the number of tasks completed by CARVS on \texttt{warehouse\_small} varies as the human workforce increases from 10 to 100 workers, with either 50 or 100 robots.
Task completion rises almost linearly up to about 50 workers, the configuration used in our experiments.
Beyond that point, each additional worker provides progressively smaller gains.

\begin{figure}
  \centering
  \includegraphics[width=0.95\linewidth]{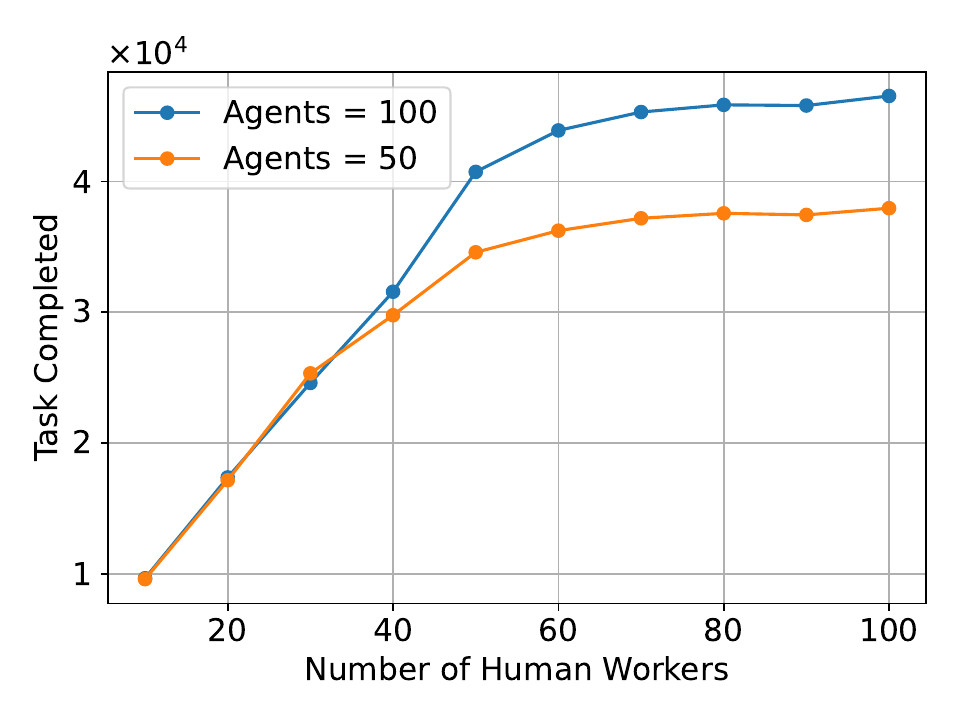}
  \caption{Impact of human workforce size on CARVS performance on \texttt{warehouse\_small} with 50 and 100 robots. After about 50 workers, each additional worker yields progressively smaller increases in the number of completed tasks.}
  \Description{Impact of human workforce size on CARVS performance on \texttt{warehouse\_small} with 50 and 100 robots. After about 50 workers, each additional worker yields progressively smaller increases in the number of completed tasks.}
  \label{fig:human_num}
\end{figure}

\section{Data and Ethics Statement}

The warehouse layouts, task records, and facility photographs used in this work were obtained through an industrial collaboration. The study involves no human subjects and no personally identifiable information. The warehouse-management-system records consist solely of stock-keeping-unit (SKU) task events over time rather than data about individuals. Each SKU appears only as an index number that distinguishes one item from another, not as its real identifier. The task sets used in our experiments are aggregate temporal distributions together with synthetically generated tasks derived from those distributions, not raw operational records. The photographs contain no people and no product information, because the product imagery in them has been pixelated. No personally identifiable information was collected, processed, or released as part of this work.

\end{document}